\documentclass[a4paper,11pt]{article}
\usepackage{xcolor}				

\usepackage[normalem]{ulem}

\usepackage{amssymb,latexsym}
\usepackage{euscript}
\usepackage{amsmath}
\usepackage{mathrsfs}
\usepackage{tikz} \usetikzlibrary{positioning,shapes.geometric}
\usepackage{tikz-cd}
\usepackage{bussproofs}
\usepackage{bm}
\usepackage{enumerate}
\usepackage{multicol}
\usepackage[ruled,vlined]{algorithm2e}
\usepackage{float}
\usepackage[all]{xy}
\usepackage{tikz-cd}

\def\ssi{\leftrightarrow}

\def\imp{\rightarrow}

\def\A{\mathcal A}
\def\B{\mathcal B}
\def\I{\mathcal I}
\def\C{\mathcal C}

\def\D{\mathcal D}

\newtheorem{definition}{Definition}[section]
\newtheorem{theorem}[definition]{Theorem}
\newtheorem{lemma}[definition]{Lemma}
\newtheorem{proposition}[definition]{Proposition}
\newtheorem{remark}[definition]{Remark}
\newtheorem{remarks}[definition]{Remarks}
\newtheorem{example}[definition]{Example}
\newtheorem{corollary}[definition]{Corollary}

\newenvironment{proof}{\noindent\bf Proof. \rm}{\hfill $\blacksquare$}

\newcounter{defcounter}
\newenvironment{myequation}{%
\addtocounter{equation}{-1}
\refstepcounter{defcounter}

\begin{equation}}
{\end{equation}}

\newcommand{\mm}{\mathcal} 

\def\um{\mathbf{\frac 1 2}}
\def\0{\mathbf{0}}
\def\1{\mathbf{1}}

\newcommand{\lfi}{{\bf LFI}}
\newcommand{\lfis}{{\bf LFI}s}
\newcommand{\cpl}{{\bf CPL}}
\newcommand{\iplp}{${\bf IPL}^+$}
\newcommand{\mbc}{{\bf mbC}}
\newcommand{\mbcciw}{{\bf mbCciw}}
\newcommand{\cbr}{{\bf Cbr}}
\newcommand{\cie}{{\bf Cie}}
\newcommand{\rcbr}{{\bf RCbr}}
\newcommand{\rcie}{{\bf RCie}}
\newcommand{\rci}{{\bf RCi}}
\newcommand{\rmbc}{{\bf RmbC}}
\newcommand{\rmbcciw}{{\bf RmbCciw}}

\newcommand{\cons}{\ensuremath{{\circ}}}
\newcommand{\sneg}{\ensuremath{{\sim}}}
\newcommand{\bneg}{\ensuremath{{-}}}

\begin{document}

\title{Paraconsistent Belief Revision: A Replacement-Enriched LFI for Epistemic Entrenchment}

\author{
  Marcelo E. Coniglio $^\textup{\scriptsize a,b}$
    \and
    Mart\'{\i}n Figallo
  $^\textup{\scriptsize c}$
    \and
    Rafael R. Testa
  $^\textup{\scriptsize a}$
}

\date{
    $^\textup{\scriptsize a}$\textit{\small Centre for Logic, Epistemology and the History of Science (CLE), State University of Campinas (Unicamp), Campinas, Brazil}
    \\
    $^\textup{\scriptsize b}$\textit{\small Institute of Philosophy and the Humanities (IFCH), Unicamp, Campinas, Brazil}
        \\
    $^\textup{\scriptsize c}$\textit{\small Departamento de Matem\'atica. Universidad Nacional del Sur. Bah\'{\i}a Blanca, Argentina}
   \\
   {\small e-mails: \texttt{coniglio@unicamp.br}; \ \  \texttt{figallomartin@gmail.com}; \ \  \texttt{rafaeltesta@gmail.com}}
}

\maketitle

\begin{abstract}
We further develop the formal foundations of Paraconsistent Belief Revision (PBR) by introducing Logics of Formal Inconsistency (\lfis) specifically designed to support the development of epistemic entrenchment-based models for belief change.  The interpretation of formal consistency -- and, more broadly, of paraconsistency -- in terms of the epistemic attitudes adopted by rational agents and of these agents reasoning with potentially contradictory yet non-trivial epistemic states, respectively, is already well-established within the literature on PBR based on LFIs. However, previous approaches faced a key limitation: the absence of replacement in most LFIs prevented the construction of entrenchment-based operations. We address this gap by first revisiting and systematizing core properties essential for such modeling, formalizing them within \cbr, a previously introduced logic whose foundational properties we now examine and develop in depth. Building on this, we introduce \rcbr, a replacement-enriched, self-extensional extension of \cbr, which makes it possible -- within an \lfi-based framework -- to formally define epistemic entrenchment and to construct entrenchment-based belief revision mechanisms. This development enables a fully constructive approach to Belief Revision in paraconsistent settings, further advancing the theoretical treatment of \lfis\ and paraconsistency within the broader landscape of epistemic states and belief dynamics.
\end{abstract}

\section{Introduction}
The \emph{Logics of Formal Inconsistency} (\lfis) provide a robust framework for addressing contradictions in logical systems without collapsing into triviality \cite{car:con:16}. Mostly advanced in \cite{car.con.mar.2007}, \lfis\ form a specific family of supraclassical paraconsistent logics that serve as an umbrella term encompassing various well-established systems in the literature. These logics allow the introduction of a formal \emph{consistency operator} ``$\circ$'' to distinguish between consistent and inconsistent sentences, which plays a crucial role by ensuring that certain formulas behave classically (as clarified below), while still allowing contradictions to coexist elsewhere in the system. This ability to localize and manage contradictions is particularly valuable in dynamic environments, where agents must process potentially conflicting information without compromising the integrity of their belief systems \cite{Testa2014}.

Two systems exist in the literature of Paraconsistent Belief Revision (PBR) based on \lfis, namely AGMp and AGM$\circ$~\cite{tes.con.rib.2016}. These systems extend the classical AGM framework~\cite{agm1985} to paraconsistent settings, adapting its principles to handle contradictions without trivializing the belief set. While AGMp builds on the AGM-compliance of \lfis, thereby preserving the fundamental constructions of the original belief dynamics model with only minor reinterpretations and refinements (within the broader context of Belief Revision in non-classical logics~\cite{wassermann,ribeiro}), AGM$\circ$ explicitly incorporates the consistency operator in the constructions by introducing new epistemic attitudes -- namely, {\em strong acceptance} (both $\circ\alpha$  and $\alpha$ belong to the belief set $K$) and {\em strong rejection}  (both $\circ\alpha$  and $\neg\alpha$ belong to the belief set $K$) of a belief-representing sentence $\alpha$ (alongside the usual epistemic attitudes). In AGM$\circ$, the consistency of a sentence $\alpha$, expressed by $\circ\alpha$, ensures that $\alpha$, when accepted, cannot be retracted from the belief set through contraction unless $\circ\alpha$ itself is first removed. This mechanism guarantees the protection of strongly accepted sentences -- those that are both accepted and consistent -- thereby structuring the belief revision process around the preservation of core beliefs. Despite this innovation, AGMp and AGM$\circ$ do not address the ranking of beliefs based on their resistance to change, as encapsulated in the concept of \emph{epistemic entrenchment}.

A symmetry exists between the underlying concept of irrevocability of strongly accepted sentences in AGM$\circ$ and the notion of \emph{entrenchment}, where strongly accepted sentences could naturally correspond to the most entrenched beliefs. However, this symmetry has remained unexplored due to the inability to define \emph{epistemic entrenchment} in PBR systems based on \lfis. This limitation stems from the absence of the \emph{replacement property} in most \lfis. This property guarantees that if two formulas are logically equivalent, one can replace the other in any context without affecting derivability or inferential behavior. From this, the lack of other expected properties--such as the preservation of equivalences involving the consistency operator--hinders the formulation of a coherent framework for \emph{epistemic entrenchment} in  paraconsistent AGM systems based on \lfis.

To address these challenges, this paper introduces two paraconsistent logics in the family of \lfis: \textbf{Cbr} and \textbf{RCBr}. The logic \textbf{Cbr} (already defined in~\cite{tes.con.rib.2016}, where its main properties were presented but without a systematic or dedicated treatment) is specifically designed to support belief revision in AGM$\circ$ by incorporating properties that are essential for the formal characterization of epistemic attitudes such as strong acceptance and strong rejection. In particular, the consistency operator $\circ$ satisfies $\circ \alpha \equiv \circ \neg \alpha$, reflecting a symmetry in the evaluation of consistency for a formula and its negation.  Since $\alpha \equiv \neg\neg\alpha$ in  \textbf{Cbr}, it follows that, for every sentence $\alpha$ and belief set $K$:   $\alpha$ is strongly accepted in $K$  if and only if $\neg\alpha$ is strongly rejected in $K$; and  $\alpha$ is strongly rejected in $K$ if and only if $\neg\alpha$ is strongly accepted in $K$. Both properties concerning epistemic attitudes are very natural and expected. Another interesting feature of   \textbf{Cbr} is that $\circ$ preserves logical equivalence under certain reasonable assumptions, namely: if $\alpha \equiv \beta$ and $\neg \alpha \equiv \neg \beta$, then $\circ \alpha \equiv \circ \beta$. These properties ensure that $\circ$ behaves coherently with respect to equivalences and provides a suitable foundation for modeling strong epistemic attitudes in belief revision systems.

Building on \cbr, we further introduce \rcbr, an extension that incorporates the \emph{replacement property} while maintaining the paraconsistent nature of the system. By resolving the limitations of previous \lfis, \rcbr\ allows for the substitution of logically equivalent formulas, supporting the construction of robust \emph{epistemic entrenchment} frameworks.

\section{Formal Background}\label{background}

As is usual in Belief Revision theory, little is required from the underlying language. Classically, the framework assumes a formal language closed under the usual truth-functional connectives. Additionally, we consider a consequence relation $Cn$ that satisfies the standard Tarskian properties:
\begin{itemize}
    \item \textbf{Reflexivity:} $\varphi \in Cn(\Gamma)$ for any $\varphi \in \Gamma$.
    \item \textbf{Monotonicity:} If $\varphi \in Cn(\Gamma)$, then $\varphi \in Cn(\Gamma \cup \Delta)$ for any $\Delta$.
    \item \textbf{Transitivity:} If $\psi \in Cn(\Gamma)$ and $\varphi \in Cn(\Gamma \cup \{\psi\})$, then $\varphi \in Cn(\Gamma)$.
\end{itemize}

In addition to the Tarskian properties, the consequence relation $Cn$ is assumed to satisfy the following additional properties:
\begin{itemize}
    \item \textbf{Supraclassicality:} If $\Gamma \models \varphi$ in classical logic, then $\varphi \in Cn(\Gamma)$.
    \item \textbf{Compactness:} If $\varphi \in Cn(\Gamma)$, then there exists a finite subset $\Gamma_0 \subseteq \Gamma$ such that $\varphi \in Cn(\Gamma_0)$.
    \item \textbf{Deduction:} For any $\Gamma$ and formulas $\varphi$ and $\psi$, $\psi \in Cn(\Gamma \cup \{\varphi\})$ if and only if $\varphi \to \psi \in Cn(\Gamma)$.
\end{itemize}

As noted by the AGM trio \cite{agm1985}, the underlying logical framework can -- and indeed should -- be extended or adapted according to the intended applications of the belief revision system. This flexibility is especially important given that belief states in the AGM framework are represented by logically closed sets -- a useful but idealized abstraction, since real agents may not always infer all logical consequences of their beliefs. Moreover, the choice of the consequence operator $Cn$ does not commit one to a specific logic: although it typically contains classical truth-functional logic, it may include additional principles (as it is observed by \cite{ferme.hansson}), depending on the epistemological and inferential goals of the belief change system.

In line with this perspective, we introduce a specific \lfi\ designed to address the challenges inherent in paraconsistent belief revision. This logic is presented semantically via a non-deterministic matrix (Nmatrix) framework, which enhances its expressive capabilities and yields a straightforward decision procedure. The formal development of this \lfi, including its semantic foundations and key properties, is carried out throughout the article.

\section{Some \lfis\ suitable for paraconsistent  belief revision}
The present section introduces the formal systems that serve as the logical foundations for our belief revision framework. While all definitions, results, and semantic tools required for the developments in this paper are self-contained and explicitly stated here, we refer the interested reader to \cite{car:con:16} and \cite{car.con.mar.2007} for a broader and more detailed exposition of the theory of \lfis). These works provide both a historical and technical overview of \lfis, including their algebraic semantics, axiomatic families, and motivations from the standpoint of paraconsistency. Here, we restrict ourselves to the essential formal components needed to define and analyze the logics \cbr\ and \rcbr, beginning with their syntactic foundations and basic semantic structure.

Consider the propositional signature $\Sigma=\{\vee, \wedge, \rightarrow, \neg, \circ\}$ and denote by $\mathfrak{Fm}$ the algebra of formulas generated by a denumerable set of propositional variables over the signature $\Sigma$. As usual, $\alpha\ssi\beta$ is an abbreviation of $(\alpha\imp\beta)\wedge(\beta\imp\alpha)$.  Let \cbr\ be the logic given by the folowing Hilbert-style calculus:

\

{\bf Axioms}

\vspace{-0.3cm}
\begin{myequation}
  \alpha\imp(\beta\imp\alpha)
   \end{myequation}
\vspace{-0.5cm}
\begin{myequation}
  \big(\alpha\imp(\beta\imp\gamma)\big)\imp\big((\alpha\imp\beta)\imp(\alpha\imp\gamma)\big)
   \end{myequation}
\vspace{-0.5cm}
\begin{myequation}
  \alpha\imp(\beta \imp (\alpha\wedge\beta))
   \end{myequation}
\vspace{-0.5cm}
\begin{myequation}
  (\alpha\wedge\beta)\imp\alpha
   \end{myequation}
\vspace{-0.5cm}
\begin{myequation}
  (\alpha\wedge\beta)\imp\beta
   \end{myequation}
\vspace{-0.5cm}
\begin{myequation}
  \alpha\imp(\alpha\vee\beta)
   \end{myequation}
\vspace{-0.5cm}
\begin{myequation}
  \beta\imp(\alpha\vee\beta)
   \end{myequation}
   \vspace{-0.5cm}
\begin{myequation}
  \big(
\alpha\imp\gamma
\big)
\imp
\big(
(\beta\imp\gamma)\imp((\alpha\vee\beta)\imp\gamma)\big
)
   \end{myequation}
\vspace{-0.5cm}
\begin{myequation}
  \big(\alpha\imp\beta)\vee\alpha
   \end{myequation}
\vspace{-0.5cm}
\begin{myequation}
\alpha\vee\neg\alpha
   \end{myequation}
\vspace{-0.5cm}
\begin{myequation}
{\circ}\alpha\imp\big(\alpha\imp(\neg\alpha\imp\beta)\big)
   \end{myequation}
\vspace{-0.5cm}
\begin{myequation}
{\circ}\alpha\vee(\alpha\wedge\neg\alpha)
   \end{myequation}
\vspace{-0.5cm}
\begin{myequation}
 \alpha\imp{\neg}{\neg}\alpha
   \end{myequation}
   \vspace{-0.5cm}
\begin{myequation}
{\neg}{\neg}\alpha\imp\alpha
   \end{myequation}

\noindent {\bf Inference Rule}

$$\mbox{(MP)} \hspace{.3cm} \displaystyle \frac{\alpha \hspace{.5cm} \alpha\imp\beta}{\beta} $$

\

\noindent
Finally, let \cie\ be the logic obtained from \cbr\ by replacing axiom ({\bf Ax}12) by the following:
\begin{myequation}
\neg{\circ}\alpha\to (\alpha\wedge\neg\alpha)
   \end{myequation}

\noindent
The notion of formal proof in \cbr\ and \cie\ is the usual one. We write $\Gamma \vdash_{\cbr} \varphi$ ($\Gamma \vdash_{\cie} \varphi$, resp.) to indicate that there is a formal proof in \cbr\ (in \cie, resp.) of $\varphi$ from the set of premises  $\Gamma$.

\begin{remarks}  \label{rems-Ci} (1) It is worth noting that axioms ({\bf Ax}1)-({\bf Ax}11) plus (MP) constitute a Hilbert calculus for \mbc, the minimal \lfi\ considered in~\cite{car.con.mar.2007}. In~\cite{car:con:16}, axiom ({\bf Ax}12) was called ({\bf cwi}), and the \lfi\ obtained by  axioms ({\bf Ax}1)-({\bf Ax}12) plus (MP) (i.e., \mbc\ + ({\bf Ax}12)) was called \mbcciw.\\[1mm]
(2) By~\cite[Proposition~3.1.10]{car:con:16} it is known that, in the presence of \mbc, axiom  ({\bf Ax}15) is equivalent to  ({\bf Ax}12) plus ({\bf cp1}): ${\circ}{\circ}\alpha$.
\end{remarks}

By using the general techniques for generation of Nmatrices as swap structures considered in~\cite{con:24}, and taking into account that \cbr\ and \cie\ can be obtained by means of the axioms considered in Section~6.4 of that paper, we get the following:

\begin{definition}\label{DefMatCbr}
Let $\mm M_{\cbr}$ be the  three-valued Nmatrix $\langle \mm T,  \mm D, \{\hat\wedge,\hat\vee,\hat\rightarrow, \hat\neg, \hat\circ \}\rangle$ over the signature $\Sigma$ with domain $\mm T=\{\1, \um, \0\}$ and set of
designated values $\mm D=\{\1, \um\}$ and  such that the truth-tables associated to each connective are the following:

\begin{center}
\begin{tabular}{| c || c | c | c | }
\hline 
$\hat\wedge$ & $\1$ & $\um$ & $\0$ \\ \hline \hline
$\1$ & $\mm D$ & $\mm D$ & $\{\0\}$\\ \hline
$\um$ & $\mm D$ & $\mm D$ & $\{\0\}$\\ \hline
$\0$ & $\{\0\}$ & $\{\0\}$ & $\{\0\}$\\ \hline
\end{tabular} \,
\begin{tabular}{| c || c | c | c | }
\hline
$\hat\vee$ & $\1$ & $\um$ & $\0$ \\ \hline \hline
$\1$ & $\mm D$ & $\mm D$ & $\mm D$\\ \hline
$\um$ & $\mm D$ & $\mm D$ & $\mm D$\\ \hline
$\0$ & $\mm D$ & $\mm D$ & $\{\0\}$\\ \hline
\end{tabular} \, 
\begin{tabular}{| c || c | c | c | }
\hline
$\hat\rightarrow$ & $\1$ & $\um$ & $\0$ \\ \hline \hline
$\1$ & $\mm D$ & $\mm D$ & $\{\0\}$\\ \hline
$\um$ & $\mm D$ & $\mm D$ & $\{\0\}$\\ \hline
$\0$ & $\mm D$ & $\mm D$ & $\mm D$\\ \hline
\end{tabular} \,
\begin{tabular}{| c || c || c |}
\hline
 & $\hat\neg$ & $\hat\circ$  \\ \hline\hline
$\1$ & $\{\0\}$ & $\D$ \\ \hline
$\um$ & $\{\um\}$ & $\{\0\}$\\ \hline
$\0$ & $\{\1\}$ & $\D$ \\ \hline
\end{tabular}
\end{center}

\

\noindent The three-valued Nmatrix $\mm M_{\cie}$ for \cie\ is obtained from $\mm M_{\cbr}$ by replacing $\hat{\cons}$ by the following:

$$\begin{tabular}{| c || c |}
\hline
 & $\hat\circ$  \\ \hline\hline
$\1$ & $\{\1\}$ \\ \hline
$\um$ & $\{\0\}$\\ \hline
$\0$ & $\{\1\}$ \\ \hline
\end{tabular}$$
\end{definition}

\noindent Let ${\bf L} \in \{\cbr, \cie\}$.
The notion of valuation is the usual in Nmatrices, namely: a function $v:\mathfrak{Fm} \to \mm T$ is a valuation in $\mm M_{\bf L}$ if $v(\alpha \# \beta) \in v(\alpha) \hat{\#} v(\beta)$ for $\# \in \{\land,\vee,\to\}$, and $v(\#\alpha) \in \hat{\#}v(\alpha)$ for $\# \in \{\neg,\circ\}$. We say that $\varphi$ is a consequence of $\Gamma$ in {\bf L},  and we write $\Gamma \models_{\mm M_{\bf L}} \varphi$, if for every valuation $v$ in the Nmatrix $\mm M_{\bf L}$ it holds: if $v(\gamma)\in \mm D$ for every $\gamma \in \Gamma$ then $v(\varphi)\in \mm D$.

\begin{theorem} (Soundness and completeness of \cbr\ and \cie\  w.r.t. Nmatrix semantics)\\ 
Let ${\bf L} \in \{\cbr, \cie\}$, and let $\Gamma\cup\{\varphi\} \subseteq \mathfrak{Fm}$. Then:
$\Gamma \vdash_{\bf L} \varphi$ \ iff \ $\Gamma \models_{\mm M_{\bf L}} \varphi$.
\end{theorem}
\begin{proof} It follows from the results given in~\cite[Section~6.4]{con:24}.
\end{proof}

\begin{remark} \label{Cbr-sublog-Cie}
Clearly,  $\cons\cons\alpha$ is not valid in \cbr: take a propositional variable $p$ and a valuation $v$ in $\mm M_{\cbr}$ such that $v(p)=\1$ and $v(\cons p)=\um$. Hence, $v(\cons\cons p)=\0$. By Remarks~\ref{rems-Ci}(2), it follows that \cbr\ is properly contained in \cie.  As we shall see in the next section, the fact that ${\circ}{\circ}\alpha$ is a theorem of \cie\ has deep consequences for the paraconsistent belief system AGM$\circ$ based on \lfis\ to be considered in this paper (see Remarks~\ref{rem-RCie}).
\end{remark}

\noindent Let {\bf L} be a logic containing \mbc. When there is no risk of confusion, if $\alpha$ and $\beta$ are formulas such that $\vdash_{\bf L} \alpha \leftrightarrow \beta$, we will write $\alpha\equiv\beta$. Then:

\begin{proposition} \label{prop-Cbr}
In \cbr\ (and so in \cie) the following holds:\\[1mm]
(i) $\circ \alpha \equiv \circ\neg\alpha$, for every formula $\alpha$.\\[1mm]
(ii) If $\alpha \equiv \beta$ and $\neg\alpha \equiv \neg\beta$ then  $\circ\alpha \equiv \circ\beta$.
\end{proposition}
\begin{proof} (i) It is immediate by using the Nmatrix $\mm M_{\cbr}$.\\[1mm]
(ii) It is an immediate consequence of the fact that in \mbcciw, the extension of \mbc\ by axiom ({\bf Ax}12), $\circ\alpha \equiv {\sim}(\alpha \land \neg\alpha)$. Here, $\sim$ is the classical negation definable in \mbcciw\ as $\sneg \alpha :=\neg\alpha \land \circ\alpha$ (see~\cite[Chapter~3]{car:con:16}).
\end{proof}

\

The properties established in Proposition \ref{prop-Cbr} express two key epistemic symmetries that have been regarded as desirable since the earliest attempts to provide an epistemic interpretation of \lfis\ through belief revision. The first asserts that the consistency of a sentence is epistemically indistinguishable from the consistency of its negation -- capturing the idea that regarding a belief as consistent is fundamentally symmetric with respect to acceptance and rejection. The second ensures that if two sentences (as well as their negations) are logically equivalent, then asserting their consistency should amount to the same epistemic commitment. These conditions support the coherent treatment of strong epistemic attitudes -- such as strong acceptance and strong rejection -- within a belief set. Indeed, in \cite{Testa2014}, such properties were already anticipated through formulations of the form ``let L be an \lfi\ satisfying\dots,'' highlighting their relevance even before a concrete logic was explicitly formulated. A more systematic treatment was later proposed in \cite{tes.con.rib.2016}, where \cbr\ was explicitly introduced, axiomatized, and had its main properties formally established. In the present paper, we revisit and systematize these results as a foundational step towards introducing \rcbr, a self-extensional extension of \cbr\ designed to capture further structural properties relevant to belief dynamics -- most notably, the replacement property (to be technically explained bellow). In \cite{tfgr}, extensionality -- as reflected at the level of revision operators -- was addressed indirectly via \emph{remainder set constructions} (adapting works on AGM revision for logics without negation \cite{description}). That paper also explicitly identified the restoration of the replacement property for \lfis\ as an open problem for advancing belief revision grounded on some ordering mechanism over belief-representing sentences, such as \emph{transitively relational partial meet constructions} and epistemic entrenchment frameworks. Here, we address this challenge by incorporating the corresponding structural behavior directly into the logic itself, ensuring that the epistemic entrenchment mechanisms -- and, consequently, belief revision -- rest on solid formal principles.

In the context of algebraic logic, a logic is said to be {\em self-extensional} if it satisfies the {\em replacement property}, namely:\\

$\mbox{if } \ \alpha_1\equiv \beta_1, \ldots, \alpha_n\equiv \beta_n \ \mbox { then } \ \varphi(\alpha_1,\ldots, \alpha_n) \equiv \varphi(\beta_1,\ldots, \beta_n)$\hspace*{3cm}(R)\\

\noindent
for every $\alpha_i, \beta_i$ and $\varphi(p_1,\ldots,p_n)$.
In~\cite{car.con.fue.2022} there were considered extensions of \lfis\ by adding rules which guarantee (R). Indeed, given an \lfi, say {\bf L}, its self-extensional version {\bf RL}  is obtained from {\bf L} by adding the (global) inference rules

$$(E_\neg) \hspace{.3cm} \displaystyle \frac{\alpha \leftrightarrow \beta}{\neg\alpha \leftrightarrow \neg\beta} \hspace{2cm}(E_\circ) \hspace{.3cm} \displaystyle \frac{\alpha \leftrightarrow \beta}{\circ\alpha \leftrightarrow \circ\beta}$$

\

The new inference rules are global, in the sense that it can only be applied to theorems. In order to recover the deduction metatheorem we adopt, as it is usually done with the necessitation rule in modal logics, the following notion of derivations in {\bf RL}:

\begin{definition} \label{def:der:RL}
We say that $\alpha$ is {\em derivable in {\bf RL}}, written $\vdash_{\bf RL} \alpha$, if there is a derivation of $\alpha$ in {\bf RL} in the usual sense. On the other hand, $\alpha$ is {\em derivable in {\bf RL} from a set of premises $\Gamma$}, written as $\Gamma \vdash_{\bf RL} \alpha$, if either $\vdash_{\bf RL} \alpha$ or there exists $\gamma_1,\ldots,\gamma_n \in \Gamma$ such that $\vdash_{\bf RL} (\gamma_1 \wedge \ldots \wedge \gamma_n) \to \alpha$. 
\end{definition}

\noindent 
The algebraic semantics for this class of self-extensional \lfis\ given in~\cite{car.con.fue.2022} consists of expansions of Boolean algebras by adding operators $\tilde{\neg}$ and $\tilde{\cons}$ for the new connectives. From now on, the Boolean operators (meet, join, complement and implication) in a given Boolean algebra $\A$  will be denoted, respectively, by $\sqcap$, $\sqcup$, $-$ and $\Rightarrow$. The minimum and maximum element will be denoted by $0$ and $1$, respectively.

\begin{definition}
A {\em Boolean algebra with \lfi\ operators} (BALFI, for short) for \rmbc\ is a  Boolean algebra $\A$ expanded with two unary operators $\tilde{\neg}$ and $\tilde{\cons}$ such that $x\sqcup \tilde{\neg}x=1$ and $x \sqcap \tilde{\neg}x \sqcap \tilde{\cons}x=0$ hold, for every $x \in \A$. A BALFI for \rmbcciw\ is a BALFI for \rmbc\ such that $\tilde{\cons} x = \bneg(x \sqcap \tilde{\neg} x)$ for every $x$. A BALFI for \rcbr\ is a BALFI for \rmbcciw\ such that $\tilde{\neg}\tilde{\neg} x = x$ for every $x$. A BALFI for \rcie\ is a BALFI for \rmbc\ such that $\tilde{\neg}\tilde{\cons} x = x \sqcap \tilde{\neg} x$ and  $\tilde{\neg}\tilde{\neg} x = x$ for every $x$. Given ${\bf L} \in \{\mbc,\mbcciw,\cbr,\cie\}$, the class of BALFIs for {\bf RL} will be denoted by $\mathbb{B}({\bf RL})$.
\end{definition}

If $\B$ is a BALFI for {\bf RL} and $h:\mathfrak{Fm} \to \mathcal{B}$ is a  homomorphism, $h$ is said to be a {\em valuation over $\B$}.

\begin{definition} \label{def:val:RL} Let {\bf RL} be as above, and let $\mathcal{B} \in \mathbb{B}({\bf RL})$.\\[1mm]
(1) We say that $\alpha$ is {\em valid in $\mathcal{B}$}, written $\models_{\mathcal{B}} \alpha$, if $h(\alpha)=1$ for every valuation $h$ over $\mathcal{B}$. A formula {\em $\alpha$ is valid in $\mathbb{B}({\bf RL})$}, written $\models_{\mathbb{B}({\bf RL})} \alpha$, if it is valid in every $\mathcal{B} \in \mathbb{B}({\bf RL})$.\\[1mm]
(2) Let $\Gamma \cup \{\alpha\} \subseteq \mathfrak{Fm}$. We say that {\em $\alpha$ is a consequence of $\Gamma$ in $\mathbb{B}({\bf RL})$}, written $\Gamma \models_{\mathbb{B}({\bf RL})} \alpha$, if either $\alpha$ is valid in $\mathbb{B}({\bf RL})$, or there exist $\gamma_1,\ldots,\gamma_n \in \Gamma$ such that $(\gamma_1 \wedge \ldots \wedge \gamma_n) \to \alpha$ is valid in $\mathbb{B}({\bf RL})$. 
\end{definition}

\begin{theorem} (Soundness and completeness of {\bf RL}  w.r.t. BALFI semantics)\\ 
Let ${\bf L} \in \{\mbc,\mbcciw,\cbr,\cie\}$, and let $\Gamma\cup\{\varphi\} \subseteq \mathfrak{Fm}$. Then:
$\Gamma \vdash_{\bf RL} \alpha$ \ iff \ $\Gamma \models_{\mathbb{B}({\bf RL})} \alpha$.
\end{theorem}
\begin{proof}
It follows by adapting the results obtained in~\cite[Sections~2 and~3]{car.con.fue.2022}.
\end{proof}

\begin{remarks} \label{rem-RCbr} (1)  In~\cite[Proposition~3.24]{car.con.fue.2022} it was shown that \rci\ also satisfies da Costa's axiom
  \begin{gather}
    \neg(\alpha \land \lnot \alpha) \to \cons\alpha                        \tag{\bf cl}
  \end{gather}
From this, the BALFIs for \rci\ (and so for \rcie) are such that $\tilde{\neg}(x \sqcap \tilde{\neg} x)=\tilde{\cons} x = \bneg(x \sqcap \tilde{\neg} x)$ for every $x$.\\[1mm]
(2) Let \rcbr\ be the self-extensional version of \cbr. Because of Proposition~\ref{prop-Cbr}(ii), in order to guarantee (R) it is enough adding to \cbr\ the rule $(E_\neg)$, since the other rule $(E_\circ)$ is a derived one. That is, $\rcbr=\cbr+(E_\neg)$,  and it satisfies: $\alpha \equiv \beta$ implies that $\neg\alpha\equiv\neg\beta$ and $\circ\alpha \equiv\circ\beta$. From this, \rcbr\ satisfies (R), given that the binary connectives $\land$, $\lor$ and $\to$ already preserve logical equivalences. This is a consequence of   axioms ({\bf Ax}1)-({\bf Ax}8) plus (MP), which constitute a Hilbert calculus for  positive intuitionistic logic \iplp, a self-extensional logic.  Observe also that  $\alpha \equiv \neg\neg\alpha$ for every sentence $\alpha$. The same properties hold for \rcie. In particular, $\rcie=\cie+(E_\neg)$. Clearly, \rcbr\ is properly contained in \rcie.
\end{remarks}

\

\noindent We know from~\cite{car.con.fue.2022} that there are paraconsistent models for \rci, and so for its sublogics  \rmbc\ and \rmbcciw. We will show now that there are paraconsistent models for \rcie\ (and so for its sublogic  \rcbr).
 This is a crucial point, since the existence of paraconsistent models for these logics, as well as a (paraconsistent) model for \rcbr\ invalidating the laws\\
 
$\begin{array}{ll}
(cp1) & \ \ \ \ \cons\cons\alpha\\[1mm]
(cp2) & \ \ \ \ \cons\alpha \to \cons\cons\alpha\\[1mm]
(cp3) & \ \ \ \ (\alpha \land\cons\alpha) \to \cons\cons\alpha\\[1mm]
(cp4) & \ \ \ \ (\neg\alpha \land\cons\alpha) \to \cons\cons\alpha\\[2mm]
\end{array}$

\noindent (see Proposition~\ref{countermodel}) show that they are, in fact, paraconsistent. More than this, the fact that (cp1)-(cp4) do not hold in general in \rcbr\ is fundamental to our purposes, as discussed in Remarks~\ref{rem-RCie}(2). That is, the proposed formal  framework has models as expected.

\begin{proposition} \label{RCie-paraconsist} \rcie\ is an \lfi. In particular, it is a paraconsistent logic. The same holds for \rcbr.
\end{proposition}
\begin{proof}
By definition of \lfis, it is enough defining a BALFI, say $\mathcal{B}$, such that, for two different propositional variables $p$ and $q$, it holds: (i) $p, \neg p \not\models_\mathcal{B} q$; (ii) $p, \cons p \not\models_\mathcal{B} q$; and (iii) $\cons p, \neg p \not\models_\mathcal{B} q$ (see~\cite[Definition~2.1.7]{car:con:16}).
To this end, consider the Boolean algebra $\A=\wp(\mathbb{Z})$ (the powerset of the set $\mathbb{Z}$ of integer numbers).  For each $n \in \mathbb{Z}$ let
$$(-\infty, n] := \{ m \in \mathbb{Z} \ : \ m \leq n \} \ \ \mbox{ and } \ \  [n,\infty) := \{ m \in \mathbb{Z} \ : \ m \geq n \}.$$
Let $\mathcal{I} := \{ (-\infty, n] \ : \ n \in \mathbb{Z} \} \cup \{ [n,\infty) \ : \ n \in \mathbb{Z} \} $ be the set of {\em inconsistent truth-values} and $\mathcal{C}:=\A \setminus \mathcal{I}$ be the set of {\em consistent truth-values}. Observe that, for every $X \in \A$: $X \in \mathcal{I}$ iff $\mathbb{Z} \setminus X \in \mathcal{I}$. From this, $X \in \mathcal{C}$ iff $\mathbb{Z} \setminus X \in \mathcal{C}$, for every $X \in \A$.
Define unary operators $\tilde{\neg}, \tilde{\circ}:\A \to \A$ as follows: $\tilde{\neg}(-\infty, n] = [n,\infty)$; $\tilde{\neg}[n,\infty) = (-\infty,n]$; and $\tilde{\neg}X=\mathbb{Z} \setminus X$, if $X \in \mathcal{C}$. By the observation above, $\tilde{\neg}$ is a well-defined function. Finally, define $\tilde{\cons}X=\mathbb{Z} \setminus (X \cap \tilde{\neg}X)$. We have the following table:

\begin{center}
\begin{tabular}{|c|c|c|c|c|c|c|}
\hline
$X$ & $\tilde{\neg} X$ & $X \cap \,\tilde{\neg}X$ & $\tilde{\cons}X$ & $\tilde{\neg}\tilde{\neg}X$ & $\tilde{\cons}\tilde{\neg}X$  & $\tilde{\neg}\tilde{\cons}X$ \\ \hline \hline
$(-\infty, n]$ & $[n,\infty)$ & $\{n\}$ & $\mathbb{Z} \setminus\{n\}$ & $(-\infty, n]$ & $\mathbb{Z} \setminus\{n\}$ & $\{n\}$\\ \hline
$[n,\infty)$ & $(-\infty, n]$ & $\{n\}$ & $\mathbb{Z} \setminus\{n\}$ & $[n,\infty)$ & $\mathbb{Z} \setminus\{n\}$ & $\{n\}$ \\ \hline
$X \in \mathcal{C}$ & $\mathbb{Z}\setminus X$ & $\emptyset$ & $\mathbb{Z}$ & $X$ & $\mathbb{Z}$ & $\emptyset$  \\ \hline
\end{tabular}
\end{center}
\noindent Clearly, the resulting structure is a BALFI $\mathcal{B}$ for \rcie, which satisfies the desired requirements. This shows that \rcie\ is an \lfi\ and, in particular, it is a paraconsistent logic. The same holds for \rcbr.
\end{proof}

\section{A fundamental  countermodel in \rcbr\ using modular arithmetic} \label{sect:countermodel}

This section is devoted  to prove that the schemas \\

$\begin{array}{ll}
(cp1) & \ \ \ \ \cons\cons\alpha\\[1mm]
(cp2) & \ \ \ \ \cons\alpha \to \cons\cons\alpha\\[1mm]
(cp3) & \ \ \ \ (\alpha \land\cons\alpha) \to \cons\cons\alpha\\[1mm]
(cp4) & \ \ \ \ (\neg\alpha \land\cons\alpha) \to \cons\cons\alpha\\[2mm]
\end{array}$

\

\noindent
are not valid in \rcbr. A standard way to show that a certain schema does not hold in a given logic is to find a model of the logic which invalidates that schema. The BALFI semantic for logic  \rcbr\ is intrincate, and the use of infinite (paraconsistent) models is essential. In the present case, an interesting BALFI for \rcbr, called $\B_{mod}$, will we introduced, showing that it invalidates the schemas (cp1)-(cp4) above. Observe that (cp1) implies (cp2), and the later implies both (cp3) and (cp4). Then, if (cp3) or (cp4) are shown to be not valid in \rcbr\ then (cp2) is not valid, therefore (cp1) is also not valid.

The  BALFI  $\B_{mod}$ will be defined as an expansion of the Boolean algebra $\wp(\mathbb{Z})$. However, defining the inconsistent truth-values (where $\tilde{\neg}$ is paraconsistent) is considerably more complex than in the BALFI over $\wp(\mathbb{Z})$ built in Proposition~\ref{RCie-paraconsist} to show that \rcie\ is paraconsistent.
Indeed, a brief analysis of the basic properties of the set of inconsistent truth-values in  \rcbr\ is necessary, in order to obtain a sufficient condition which guarantees that a BALFI satisfies the required properties. Basically, the set $\mathcal{I}$ of  inconsistent truth-values in any BALFI for \rcbr\ over  $\wp(\mathbb{Z})$ must be closed under paraconsistent negation $\tilde{\neg}$ and classical negation  (i.e., Boolean complement $X^c:=\mathbb{Z}\setminus X$). Moreover, there must be at least one element $X$ in $\mathcal{I}$ such that $X \cap  \tilde{\neg} X$ is also in $\I$ (see Proposition~\ref{inval-cp} below). Since $\tilde{\neg}$ is involutive, it is enough to define, for every  $X \in \I$, the sets  $\tilde{\neg} (X)$ and $\tilde{\neg}(X^c)$. To obtain a comprehensive yet manageable family of inconsistent sets, we define them as finite unions and intersections of equivalence classes (or their complements) of integers modulo different prime numbers. The use of (very basic) properties of modular arithmetic will be therefore used along the construction of $\I$. In special, the {\em Chinese Remainder Theorem} (CRT) will play a fundamental role in this construction, since it will guarantee that these sets (the elements of $\I$) are non-empty.

Let us start by stating some basic results on paraconsistent BALFIs for \rcbr, obtaining a sufficient condition to {\em invalidate} axiom schema  (cp2) (and so (cp1)).

In order to define a paraconsistent BALFI in general, starting from a Boolean algebra  $\A$, it is necessary to consider a subset $\emptyset \neq\mathcal{I} \subseteq \A$ formed by the {\em inconsistent elements} of $\A$. That is, $x \in \I$ iff $x \sqcap  \tilde{\neg} x \neq 0$ (where $\tilde{\neg}$ is the paraconsistent negation operator to be defined). The fact that $\I\neq\emptyset$ reflects the  paraconsistent nature of the BALFI. Let $\C=\A \setminus \I = \{x \in \A \ : \ x \sqcap \tilde{\neg} x =0\}$ be the set of {\em consistent elements} of $\A$. Clearly, $0 \in \C$ and so $0 \notin \I$. Observe that  $\I = \A \setminus \C$, and
$$(cneg) \ \ \ \ x \in \C \ \mbox{  iff } \ \tilde{\neg} x = \bneg x.$$

\noindent
It is immediate to see that the Boolean complement $\bneg$ is injective, that is,
$$(bneg) \ \ \ \ \bneg x = \bneg y \ \mbox{  iff } x=y .$$
The same holds for $\tilde{\neg}$ in any BALFI for \rcbr:

\begin{lemma} \label{lema-BALFI0}
Let $\B$ be a  BALFI for \rcbr. Then, $\tilde{\neg}$ is injective, that is: $\tilde{\neg} x = \tilde{\neg} y$ if and only if $x=y$.
\end{lemma}
\begin{proof}
Suppose that $\tilde{\neg} x = \tilde{\neg} y$. Then, $x = \tilde{\neg}\tilde{\neg} x = \tilde{\neg}\tilde{\neg} y = y$. The converse is immediate.
\end{proof}

\begin{lemma} \label{lema-BALFIs}
Let $\B$ be a BALFI for \rcbr. Then:

$\begin{array}{lll}
(1) \ x \in \C \ \mbox{ iff } \ x =\tilde{\neg}\bneg x; &\hspace{1cm}& (4) \ x \in \C \ \mbox{ iff } \ \tilde{\neg} x \in \C; \\
(2) \ x \in \C \ \mbox{ iff } \ x =\bneg\tilde{\neg} x;  && (5)  \ x \in \I \ \mbox{ iff } \ \bneg x \in \I;  \\
(3) \ x \in \C \ \mbox{ iff } \ \bneg x \in \C;  && (6) \ x \in \I \ \mbox{ iff } \ \tilde{\neg} x \in \I. 
\end{array}
$
\end{lemma}
\begin{proof} (1) $x \in \C$ iff $\tilde{\neg} x = \bneg x$, by (cneg), iff $x = \tilde{\neg}\tilde{\neg} x = \tilde{\neg}\bneg x$, by Lemma~\ref{lema-BALFI0}.\\[1mm]
(2) $x \in \C$ iff $\tilde{\neg} x = \bneg x$, by (cneg), iff $\bneg \tilde{\neg} x = \bneg\bneg x=x$, by (bneg).\\[1mm]
(3) $x \in \C$ iff $\tilde{\neg}\bneg x = x = \bneg\bneg x$, by~(1), iff  $\bneg x \in \C$, by (cneg).\\[1mm]
(4) $x \in \C$ iff $\bneg\tilde{\neg} x = x = \tilde{\neg}\tilde{\neg} x$, by~(2), iff  $\tilde{\neg} x \in \C$, by (cneg).\\[1mm]
(5) $x \in \I$ iff $x \notin \C$ iff $\bneg x \notin \C$, by~(3), iff $\bneg x \in \I$.\\[1mm]
(6) $x \in \I$ iff $x \notin \C$ iff $\tilde{\neg} x \notin \C$, by~(4), iff $\tilde{\neg} x \in \I$.
\end{proof}

\

\noindent Next result establishes a sufficient condition for a paraconsistent BALFI  not to  validate the schema (cp2).

\begin{proposition} \label{inval-cp}
Let $\B$ be a paraconsistent BALFI for \rcbr. Suppose that there exists $x \in \I$ such that  $x \sqcap  \tilde{\neg} x \in \I$. Then, $\tilde{\circ} x \not\leq \tilde{\circ} \tilde{\circ} x$ and so  $\B$ does not validate the schema  $\cons\alpha \to \cons\cons\alpha$.
\end{proposition}
\begin{proof}
Observe that, in any Boolean algebra, $a \leq \bneg b$ iff $a \Rightarrow \bneg b = \bneg a \sqcup \bneg b = \bneg(a \sqcap b) =1$ iff   $a \sqcap b =0$. From this,
$$(\neq 0) \ \ \ \ a \not\leq \bneg b \ \mbox{ if and only if } \ a \sqcap b \neq 0$$
in any Boolean algebra. Now, let $\B$ be a paraconsistent BALFI for \rcbr\ and $x \in \I$ such that  $x \sqcap  \tilde{\neg} x \in \I$. By Lemma~\ref{lema-BALFIs}(5), $\tilde{\cons} x=\bneg(x \sqcap  \tilde{\neg} x)$ belongs to $\I$ and so $\tilde{\cons} x \sqcap \tilde{\neg} \tilde{\cons} x \neq 0$. Observe that $\tilde{\cons}\tilde{\cons} x=\bneg(\tilde{\cons} x \sqcap \tilde{\neg} \tilde{\cons} x)$. Given that
$$\tilde{\cons} x \sqcap (\tilde{\cons} x \sqcap \tilde{\neg} \tilde{\cons} x) = \tilde{\cons} x \sqcap \tilde{\neg} \tilde{\cons} x \neq 0,$$
then $\tilde{\cons} x \not\leq \bneg(\tilde{\cons} x \sqcap \tilde{\neg} \tilde{\cons} x)=\tilde{\cons}\tilde{\cons} x$, by $(\neq 0)$.

The latter condition obviously shows that $\B$ invalidates the schema (cp2). Indeed, let $p$ be a propositional variable and $h$ a valuation over the BALFI $\B$ such that $h(p)=x$, where $x$ is as above. Then $h(\cons p) =\tilde{\circ} x \not\leq \tilde{\circ} \tilde{\circ} x =   h(\cons\cons p)$. From this, $h(\cons p \to \cons\cons p)\neq 1$, showing that the schema (cp2) is not validated by $\B$.
\end{proof}

\

Before constructing the desired BALFI for \rcbr, let us recall some basic notions from modular arithmetic, which will be used along the construction. Let $m \geq 1$ be a positive integer. Given integers $x,y$, they are {\em congruent modulo $m$}, denoted by $x \equiv y$ (mod $m$), if $m$ divides $x-y$, i.e., $x-y=m\cdot k$ for some integer $k$\footnote{Observe that the symbol $\equiv$ is also used in this paper to denote logical equivalence between two formulas in a given logic. In turn, the symbol $-$ is also used to denote the Boolean complement in a given Boolean algebra. The context of use of the symbols $\equiv$ and $-$ will avoid any confusions.}. This defines an equivalence relation. More than this, it is a congruence on the ring $\mathbb{Z}$ of integers, i.e.:   $x \equiv y$ (mod $m$) and  $x' \equiv y'$ (mod $m$) implies that  $x \star x' \equiv y \star y'$ (mod $m$) for $\star\in \{+,\cdot\}$. The equivalence class of $x$ (mod $m$) will be denoted by $[x]_m$, that is, $[x]_m = \{ y \in \mathbb{Z} \ : \  x \equiv y$ (mod $m$)$\}$. Clearly,   $x \equiv y$ (mod $m$) iff $x$ and $y$ have the same remainder  when divided by $m$, hence there are only $m$ equivalence classes modulo $m$: $[0]_m, [1]_m, \ldots, [m-1]_m$.

Recall that two positive integers $n_1$, $n_2$ are {\em coprime} if there is no prime number which divides simultaneously both numbers. In particular, if  $n_1$ and $n_2$ are two different prime numbers then they are coprime.
The following classical result, which appeared for the first time in the Chinese manuscript {\em  Sunzi Suanjing} (3rd to 5th century {\em AD}), will be fundamental to our purposes:

\begin{theorem} [Chinese Remainder Theorem, CRT] \label{CRT}
Let $n_1, \ldots, n_k$ be different integers greater than $1$,  which are   pairwise coprime. Let $a_1, \ldots, a_k$ be positive integers such that $a_i < n_i$  for $1 \leq i \leq k$. Then, the system of equations $x \equiv a_i$ (mod $n_i$) ($1 \leq i \leq k$)  has a solution, and any two solutions  are congruent modulo $m=n_1 \cdot \ldots \cdot n_k$.
\end{theorem}

Only prime numbers will be considered as modules in our construction, hence CRT will be applied safely. Thus, consider the standard enumeration of the prime numbers $m_1, m_2, m_3, \ldots$ such that $m_k < m_{k+1}$ for every $k$. Then, $m_1=2$, $m_2=3$, $m_3=5$, and so on. 

To simplify the presentation, consider the following notation:

\begin{enumerate}
\item  $X^c$ stands for $\mathbb{Z} \setminus X$, for $X \subseteq \mathbb{Z}$. Hence, $[x]_m^c$ stands for $\mathbb{Z} \setminus [x]_m$, for $0 \leq x < m$; i.e., $[x]_m^c =  \bigcup \{[y]_m  \ : \ 0 \leq y < m, \mbox{ and } y \neq x\}$.

\item $\bar{\mathcal{S}}$ stands for $\{X^c \ : \ X \in \mathcal{S}\}$, for $\mathcal{S} \subseteq \wp(\mathbb{Z})$.

\item For $k \geq 1$ let $[1]^c_{3,k}=[1]_3^c \cup [1]_5^c \cup \ldots \cup [1]_{m_k}^c$. Observe that  $[1]^c_{3,2}=[1]_3^c$, and  $[1]^c_{3,1}=\emptyset$.

\item For $k \geq 2$ let $[1]^c_{5,k}=[1]_5^c \cup [1]_7^c \cup \ldots \cup [1]_{m_k}^c$. Observe that  $[1]^c_{5,3}=[1]_5^c$, and  $[1]^c_{5,2}=\emptyset$.
\end{enumerate}


\begin{remark} \label{rem-compl} The following obvious property will be used intensively from now on: 
$$(inc) \ \ \ \ [0]_m \subseteq [1]_m^c, \ \mbox{ for every integer $m \geq 2$.}$$ 
Clearly, $[x]_2^c=[1-x]_2$ for $x=0,1$.
 \end{remark}

\noindent  Now, it will be defined  $\B_{mod}$, a BALFI for \rcbr\ expanding $\wp(\mathbb{Z})$, enjoying the desired properties. The construction of  $\B_{mod}$ will be made in two steps: in the first step, a family $\I_1$ of {\em inconsistent elements} in  $\wp(\mathbb{Z})$ (as discussed above) will be defined. By Lemma~\ref{lema-BALFIs},  $X \in \I_1$ iff $X^c \in \I_1$ iff $\tilde{\neg}(X) \in \I_1$. This gives origin to a BALFI structure for \rcbr. However, this is not enough to our purposes (in principle, to invalidate~(cp2)). In order to do this, and by virtue of Proposition~\ref{inval-cp}, a second family $\I_2$ of inconsistent elements, disjoint from $\I_1$, will be defined, starting from $I=X_0 \cap \tilde{\neg}X_0$ for a fixed (arbitrary) element $X_0$ of $\I_1$. The resulting BALFI for \rcbr, called $\B_{mod}$, will therefore invalidate the schema  (cp2) (and so (cp1)). More that this, it will also invalidate (cp3)-(cp4) (see Proposition~\ref{countermodel} below).

\begin{definition} 
Let  $I=[1]_2 \cap [0]_3$; hence, $I^c=[0]_2 \cup [0]^c_3$.  Consider the following subsets of $\wp(\mathbb{Z})$:\\[1mm]
(1) $\I_1^x = \big\{[x]_2 \cup [1]^c_{3,k} \ : \ k\geq 1 \big\} \cup \big\{[x]_2 \cup [1]^c_{3,k} \cup [0]_{m_{k+1}} \ : \ k\geq 1 \big\}$, for $x=0,1$;\\[1mm]
(2) $\I_1= (\I_1^0\cup \bar{\I}_1^0) \cup (\I_1^1\cup \bar{\I}_1^1)$; \\[1mm]
(3) $\I_2^0 = \big\{I \cup [1]^c_{5,k} \ : \ k\geq 2 \big\} \cup \big\{I \cup [1]^c_{5,k} \cup [0]_{m_{k+1}} \ : \ k\geq 2 \big\}$;\\[1mm]
(4) $\I_2^1 = \big\{I^c \cup [1]^c_{5,k} \ : \ k\geq 2 \big\} \cup \big\{I^c \cup [1]^c_{5,k} \cup [0]_{m_{k+1}} \ : \ k\geq 2 \big\}$;\\[1mm]
(5) $\I_2= (\I_2^0\cup \bar{\I}_2^0) \cup (\I_2^1\cup \bar{\I}_2^1)$; \\[1mm]
(6) $\I= \I_1 \cup \I_2$.
\end{definition}

\

\noindent In order to define a paraconsistent negation  $\tilde{\neg}$ on $\I$ such that $X \cup  \tilde{\neg}(X)=\mathbb{Z}$, the following must hold, for every $X \in \I$:  $\tilde{\neg}(X)= X^c \cup Y$ for some $Y\neq\emptyset$. In addition, it will required that  $\tilde{\neg}(\tilde{\neg}(X))=X$. Taking into account that $[0]_m \subseteq [1]^c_m$, the following operator is defined:

\begin{definition} 
Consider the function $\tilde{\neg}: \I \to \I$ defined as follows:\\[1mm]
(1) Over $\I_1^0 \cup \I_1^1$:
\begin{enumerate}
\item[(1.1)] $\tilde{\neg}([x]_2)=[1-x]_2 \cup [0]_3$, for $x=0,1$;
\item[(1.2)] $\tilde{\neg}([x]_2 \cup [1]^c_{3,k})= ([x]_2 \cup [1]^c_{3,k-1} \cup [0]_{m_k})^c$, for $x=0,1$ and $k \geq 2$;
\item[(1.3)] $\tilde{\neg}([x]_2 \cup [1]^c_{3,k} \cup [0]_{m_{k+1}})= ([x]_2 \cup [1]^c_{3,k})^c$, for $x=0,1$ and $k \geq 1$.
\end{enumerate}
(2) Over $\I_2^0$:
\begin{enumerate}
\item[(2.1)] $\tilde{\neg}(I)=I^c \cup [0]_5$;
\item[(2.2)] $\tilde{\neg}(I \cup [1]^c_{5,k})= (I \cup [1]^c_{5,k-1} \cup [0]_{m_k})^c$, for $k \geq 3$;
\item[(2.3)] $\tilde{\neg}(I \cup [1]^c_{5,k} \cup [0]_{m_{k+1}})= (I \cup [1]^c_{5,k})^c$, for $k \geq 2$.
\end{enumerate}
(3) Over $\I_2^1$:
\begin{enumerate}
\item[(3.1)] $\tilde{\neg}(I^c)=I \cup [0]_{5}$;
\item[(3.2)] $\tilde{\neg}(I^c \cup [1]^c_{5,k})=(I^c \cup [1]^c_{5,k-1} \cup [0]_{m_k})^c$, for $k \geq 3$;
\item[(3.3)] $\tilde{\neg}(I^c \cup [1]^c_{5,k} \cup [0]_{m_{k+1}})= (I^c \cup [1]^c_{5,k})^c$, for $k \geq 2$.
\end{enumerate}
(4) Finally, if  $\tilde{\neg} x=y$ then define  $\tilde{\neg} y = x$.\\[1mm] 
This completes the definition of  $\tilde{\neg}$.
\end{definition}

\

\noindent Let  $X_0=[1]_2$. Observe that $X_0 \in \I_1$ such that $\tilde{\neg}X_0 = [0]_2 \cup  [0]_{3}$. Hence, $I=X_0 \cap \tilde{\neg}X_0$ generates the family $\I_2$, while $[1]_2$ and $[0]_2=[1]^c_2$ generate $\I_1$. This will become clearer after the following Remark.

\begin{remark} \label{explaining}
Let us analyze the construction of the previous Definition in a more  explicit way:\\[1mm]
(a) Observe that  $\I_1^x=\I_1^{x,1} \cup \I_1^{x,2}$ (for $x=0,1$), where
$$\I_1^{x,1}=\big\{[x]_2, \, [x]_2 \cup [1]^c_{3}, \, [x]_2 \cup [1]^c_{3} \cup [1]^c_{5}, \, [x]_2 \cup [1]^c_{3} \cup [1]^c_{5} \cup [1]^c_{7}, \ldots\big\}$$
and
$$\I_1^{x,2}=\big\{[x]_2 \cup [0]_{3}, \, [x]_2 \cup [1]^c_{3} \cup [0]_{5}, \, [x]_2 \cup [1]^c_{3} \cup [1]^c_{5} \cup [0]_{7}, \ldots\big\}.$$
The negation operator $\tilde{\neg}$ is a function (a bijection, as we will see) from $\I_1$ to $\I_1$ such that 
$$\tilde{\neg}([x]_2)= [1-x]_2 \cup [0]_{3}, \ \ \tilde{\neg}([x]_2 \cup [0]_{3})= [1-x]_2;$$ 
it maps $\I_1^{x,1} \setminus\{[x]_2\}$ into $\bar{\I}_1^{x,2}$ (and vice versa) as follows:
$$\tilde{\neg}([x]_2 \cup [1]^c_{3})=([x]_2 \cup [0]_{3})^c, \ \ \tilde{\neg}([x]_2 \cup [1]^c_{3} \cup [1]^c_{5})=([x]_2\cup [1]^c_{3} \cup [0]_{5})^c, \ldots;$$
and it maps $\I_1^{x,2} \setminus\{[x]_2 \cup [0]_{3}\}$ into $\bar{\I}_1^{x,1}$ (and vice versa) as follows:
$$\tilde{\neg}([x]_2 \cup [1]^c_{3} \cup [0]_{5})=([x]_2 \cup [1]^c_{3})^c, \ \ \tilde{\neg}([x]_2 \cup [1]^c_{3} \cup [1]^c_{5} \cup [0]_{7})=([x]_2 \cup [1]^c_{3} \cup [1]^c_{5})^c, \ldots \, .$$
This shows that $\tilde{\neg}$ maps $\I_1^{x} \setminus\{[x]_2,[x]_2 \cup [0]_{3}\}$ into $\bar{\I}_1^{x}$ (and vice versa), and $\tilde{\neg}([x]_2), \,\tilde{\neg}([x]_2 \cup [0]_{3})\in \I_1^{1-x}$. Hence, $\tilde{\neg}$ maps $\I_1$ into $\I_1$.\\[1mm]
(b) Observe that  $\I_2^0=\I_2^{0,1} \cup \I_2^{0,2}$, where
$$\I_2^{0,1}=\big\{I, \, I \cup [1]^c_{5}, \, I \cup [1]^c_{5} \cup [1]^c_{7}, \, I \cup [1]^c_{5} \cup [1]^c_{7} \cup [1]^c_{11}, \ldots\big\}$$
and
$$\I_2^{0,2}=\big\{I \cup [0]_{5}, \, I \cup [1]^c_{5} \cup [0]_{7}, \, I \cup [1]^c_{5} \cup [1]^c_{7} \cup [0]_{11}, \ldots\big\}.$$
(c) Note that  $\I_2^1=\I_2^{1,1} \cup \I_2^{1,2}$, where
$$\I_2^{1,1}=\big\{I^c, \, I^c \cup [1]^c_{5}, \, I^c \cup [1]^c_{5} \cup [1]^c_{7}, \, I^c \cup [1]^c_{5} \cup [1]^c_{7} \cup [1]^c_{11}, \ldots\big\}$$
and
$$\I_2^{1,2}=\big\{I^c \cup [0]_{5}, \, I^c \cup [1]^c_{5} \cup [0]_{7}, \, I^c \cup [1]^c_{5} \cup [1]^c_{7} \cup [0]_{11}, \ldots\big\}.$$
The negation operator $\tilde{\neg}$ is a function (a bijection, as we will see) from $\I_2$ to $\I_2$ such that 
$$\tilde{\neg}(I)= I^c \cup [0]_{5}, \ \ \tilde{\neg}(I^c \cup [0]_{5})= I;$$
$$\tilde{\neg}(I^c)= I \cup [0]_{5}, \ \ \tilde{\neg}(I \cup [0]_{5})= I^c;$$
it maps $\I_2^{0,1} \setminus\{I\}$ into $\bar{\I}_2^{0,2}$ (and vice versa) as follows:
$$\tilde{\neg}(I \cup [1]^c_{5})=(I \cup [0]_{5})^c, \ \ \tilde{\neg}(I \cup [1]^c_{5} \cup [1]^c_{7})=(I\cup [1]^c_{5} \cup [0]_{7})^c, \ldots;$$
it maps $\I_2^{0,2} \setminus\{I \cup [0]_{5}\}$ into $\bar{\I}_2^{0,1}$ (and vice versa) as follows:
$$\tilde{\neg}(I \cup [1]^c_{5} \cup [0]_{7})=(I \cup [1]^c_{5})^c, \ \ \tilde{\neg}(I \cup [1]^c_{5} \cup [1]^c_{7} \cup [0]_{11})=(I \cup [1]^c_{5} \cup [1]^c_{7})^c, \ldots$$
(and so it maps $\I_2^{0}$ into $\bar{\I}_2^{0}$ and vice versa); it maps $\I_2^{1,1} \setminus\{I^c\}$ into $\bar{\I}_2^{1,2}$ (and vice versa) as follows:
$$\tilde{\neg}(I^c \cup [1]^c_{5})=(I^c \cup [0]_{5})^c, \ \ \tilde{\neg}(I^c \cup [1]^c_{5} \cup [1]^c_{7})=(I^c\cup [1]^c_{5} \cup [0]_{7})^c, \ldots;$$
and it maps $\I_2^{1,2} \setminus\{I^c \cup [0]_{5}\}$ into $\bar{\I}_2^{1,1}$ (and vice versa) as follows:
$$\tilde{\neg}(I^c \cup [1]^c_{5} \cup [0]_{7})=(I^c \cup [1]^c_{5})^c, \ \ \tilde{\neg}(I^c \cup [1]^c_{5} \cup [1]^c_{7} \cup [0]_{11})=(I^c \cup [1]^c_{5} \cup [1]^c_{7})^c, \ldots$$
(and so it maps $\I_2^{1}$ into $\bar{\I}_2^{1}$ and vice versa). This shows that $\tilde{\neg}$ is a function from $\I_2$ to $\I_2$.
\end{remark}

\begin{proposition} \label{B-mod-well-def}
The following holds:\\[1mm]
(1) $\I_1 \cap \I_2 = \emptyset$;\\[1mm]
(2) The function $\tilde{\neg}: \I \to \I$ is well-defined;\\[1mm] 
(3) $X \cup \tilde{\neg}X = \mathbb{Z}$, for every $X \in \I$;\\[1mm]
(4) $X \cap \tilde{\neg}X \neq \emptyset$, for every $X \in \I$;\\[1mm]
(5) $X \in \I$ iff $\tilde{\neg}X \in \I$ iff $X^c \in \I$;\\[1mm]
(6) $\tilde{\neg}\tilde{\neg}X = X$;\\[1mm]
(7) There exists $X \in \I$ such that $X \cap \tilde{\neg}X \in \I$.
\end{proposition}
\begin{proof} Straightforward, by using CRT.
\end{proof}

\

\noindent Let $\C:= \wp(\mathbb{Z}) \setminus\I$. By Proposition~\ref{B-mod-well-def}(5), $X \in \C$ iff $X^c \in \C$.

\begin{definition} Let $\B_{mod}$ be the expansion of the Boolean algebra $\wp(\mathbb{Z})$ by the unary operators $\hat{\neg},\hat{\cons}:\wp(\mathbb{Z}) \to \wp(\mathbb{Z})$ defined as follows:
\begin{itemize}
\item[(i)] $\hat{\neg}X=\tilde{\neg}X$ for $X \in \I$, and $\hat{\neg}X=X^c$ for $X \in \C$;
\item[(ii)]  $\hat{\cons}X=(X \cap \hat{\neg}X)^c$ for $X \in \wp(\mathbb{Z})$.
\end{itemize}
\end{definition}

\


\begin{proposition} \label{countermodel}
$\B_{mod}$ is a paraconsistent BALFI for \rcbr\ that does not validate the schemas (cp1)-(cp4).
\end{proposition}
\begin{proof}
Recall, from the construction of $\B_{mod}$  described above, that, for $X=[1]_2$, it holds that $\tilde{\neg}X = [0]_2 \cup  [0]_{3}$. Since both $X$ and  $I:=X \cap \tilde{\neg}X = [1]_2 \cap [0]_3$ belong to $\mathcal{I}$, it follows by Proposition~\ref{inval-cp} that $\B_{mod}$ is a paraconsistent BALFI for \rcbr\ that does not validate the schema (cp2). Now, observe that \\

$\begin{array}{ll}
(1) & \ \ \ \ \tilde{\circ}X = I^c=[0]_2 \cup  [0]_{3}^c.\\[2mm]
\end{array}$

\noindent
Moreover: $\tilde{\neg}\tilde{\circ}X = \tilde{\neg} I^c = I \cup  [0]_{5}$. From this, $\tilde{\circ}X \cap \tilde{\neg}\tilde{\circ}X = I^c \cap (I \cup  [0]_{5}) = I^c \cap [0]_{5}$. This implies that \\

$\begin{array}{ll}
(2) & \ \ \ \ \tilde{\circ}\tilde{\circ}X = (I^c \cap [0]_{5})^c=  I \cup [0]_{5}^c =([1]_2 \cap [0]_3)\cup  [0]_{5}^c \neq \mathbb{Z}.\\[2mm]
\end{array}$

\noindent Therefore, $\B_{mod}$ does not validate the schema (cp1). Moreover,\\

$\begin{array}{ll}
(3) & \ \ \ \ X \cap \tilde{\circ}X = [1]_2 \cap ([0]_2 \cup  [0]_{3}^c)= [1]_2 \cap  [0]_{3}^c\\[2mm]
(4) & \ \ \ \ \tilde{\neg}X \cap \tilde{\circ}X = ([0]_2 \cup  [0]_{3}) \cap ([0]_2 \cup  [0]_{3}^c)= [0]_2.\\[2mm]
\end{array}$

\noindent Thus, using the CRT (Theorem~ \ref{CRT}), consider first $x \in \mathbb{Z}$ such that  
$$x \equiv 1 \ \mbox{(mod $2$)}, \hspace{7mm} x \equiv 1 \ \mbox{(mod $3$)}, \hspace{7mm} x \equiv 0 \ \mbox{(mod $5$)}.$$
By~(2) and~(3), it is immediate to see that $x \in   (X \cap \tilde{\circ}X)$ but $x \notin  \tilde{\circ}\tilde{\circ}X$. Hence, $X \cap \tilde{\circ}X\not\subseteq  \tilde{\circ}\tilde{\circ}X$.  Of course, this implies that  $\tilde{\circ}X\not\subseteq  \tilde{\circ}\tilde{\circ}X$.

Using the CRT again, let $y \in \mathbb{Z}$ such that  
$$y \equiv 0 \ \mbox{(mod $2$)}, \hspace{7mm} y \equiv 0 \ \mbox{(mod $5$)}.$$
By~(2) and~(4), $y \in   (\tilde{\neg}X \cap \tilde{\circ}X)$ but $y \notin  \tilde{\circ}\tilde{\circ}X$. This shows that $\tilde{\neg}X \cap \tilde{\circ}X\not\subseteq  \tilde{\circ}\tilde{\circ}X$.

From the above considerations, let $p$ be a propositional variable and $h$ a valuation over the BALFI $\B_{mod}$ such that $h(p)=X$. Then, $h(p \land \cons p) =X \cap \tilde{\circ}X\not\subseteq  \tilde{\circ}\tilde{\circ}X =   h(\cons\cons p)$. From this, $h((p \land \cons p) \to \cons\cons p)\neq 1$, showing that the schema (cp3) is not valid in $\B_{mod}$. Analogously, $h(\neg p \land \cons p) =\tilde{\neg}X \cap \tilde{\circ}X\not\subseteq  \tilde{\circ}\tilde{\circ}X =   h(\cons\cons p)$. Hence, $h((\neg p \land \cons p) \to \cons\cons p)\neq 1$, showing that the schema (cp4) is not valid in $\B_{mod}$. Notice that $h(\cons p) =\tilde{\circ}X\not\subseteq  \tilde{\circ}\tilde{\circ}X =   h(\cons\cons p)$, and so $h(\cons p \to \cons\cons p)\neq 1$. Clearly,  $h(\cons\cons p)\neq 1$.
\end{proof}

\begin{corollary} \label{cp-not-valid}
None of the schemas (cp1)-(cp4) is valid  in \rcbr, and so they are not derivable in \rcbr, that is:\\

$\begin{array}{ll}
(1) & \ \ \ \ \nvdash_{\rcbr} \cons\cons\alpha \ \mbox{ for some formula $\alpha$;}\\[1mm]
(2) & \ \ \ \ \nvdash_{\rcbr}\cons\alpha \to \cons\cons\alpha \ \mbox{ for some formula $\alpha$;}\\[1mm]
(3) & \ \ \ \ \nvdash_{\rcbr}(\alpha \land\cons\alpha) \to \cons\cons\alpha \ \mbox{ for some formula $\alpha$;}\\[1mm]
(4) & \ \ \ \ \nvdash_{\rcbr}(\neg\alpha \land\cons\alpha) \to \cons\cons\alpha \ \mbox{ for some formula $\alpha$.}\\[2mm]
\end{array}$
\end{corollary}

\section{The AGM$\circ$ system: Contraction}\label{agm}

The AGM model of belief revision, introduced by Alchourrón, Gärdenfors, and Makinson \cite{agm1985}, formalizes how agents may change their beliefs in response to new information. In that framework, a \emph{belief set} $K$ is deductively closed: whenever $K \vdash_{L} \varphi$ (in a fixed background logic $L$, typically classical logic), then $\varphi \in K$. In this context, a belief-representing sentence $\alpha$ is \emph{accepted} in $K$ if $\alpha \in K$, \emph{rejected} if $\neg\alpha \in K$, and \emph{indeterminate} otherwise. These are the standard \emph{epistemic attitudes} an agent can hold with respect to a belief set.

Furthermore, in such setting a collection of postulates is defined to characterize the distinct operations over a belief set, namely: \textbf{expansion} ($K+\alpha$) -- the addition of a sentence $\alpha$ to a belief set $K$ that does not contradict it, so that nothing needs to be removed; \textbf{revision} ($K*\alpha$) -- the addition of a sentence that may contradict the current beliefs, requiring in this case the removal of some of them; and \textbf{contraction} ($K\div\alpha$) -- the removal of a sentence from a belief set without introducing its negation. 

These operations are constrained by general \emph{rationality criteria}, such as \emph{closure}, \emph{success}, and \emph{minimal change} -- which are made precise by associating each operation with a set of postulates. Alternatively, one may define the operations by an explicit construction. In this paper we focus on \emph{contraction}, presenting it both axiomatically and via an \emph{epistemic entrenchment} construction. This is usually sufficient to capture AGM belief dynamics, since \emph{expansion} is given by the familiar set-theoretic clause $K+\alpha = Cn(K\cup\{\alpha\})$, and \emph{revision} is definable from expansion and contraction via the \emph{Levi identity}: $K*\alpha = (K\div \neg\alpha) + \alpha$.

Specifically, in this paper we will consider contraction within the AGM$\circ$ framework -- an extension of classical AGM tailored to \lfis. Besides accommodating additional epistemic attitudes, AGM$\circ$ also admits alternative constructions for revision, in which temporarily contradictory epistemic states may occur (and be methodologically exploited) as intermediate stages of belief change.

As for epistemic attitudes, in addition to \emph{acceptance}, \emph{rejection}, and \emph{indeterminacy}, AGM$\circ$ allows the following further classifications for a belief-representing sentence $\alpha$ (relative to a belief set $K$): $\alpha$ is \emph{overdetermined} (or \emph{contradictory}) if $\alpha\in K$ and $\neg\alpha\in K$; $\alpha$ is \emph{consistent} if $\circ\alpha\in K$ (independently of whether $\alpha$ is accepted or rejected); $\alpha$ is \emph{strongly} (or \emph{boldly}) \emph{accepted} if $\alpha\in K$ and $\circ\alpha\in K$; and $\alpha$ is \emph{strongly} (or \emph{boldly}) \emph{rejected} if ${\sim}\alpha\in K$, where (in extensions of) \mbcciw  , ${\sim}\alpha := \neg\alpha \land \circ\alpha$ (i.e., $\neg\alpha\in K$ and $\circ\alpha\in K$).

It is worth emphasizing that, in AGM$\circ$, a sentence may be \emph{accepted} and \emph{rejected} simultaneously -- namely, when $\alpha\in K$ and $\neg\alpha\in K$ -- without producing the trivial belief set (this only happens if $\alpha$ is also \emph{consistent} in $K$). Indeed, the classical explosive behaviour is preserved for \emph{strongly accepted} and \emph{strongly rejected} sentences (see Remark~\ref{rem-RCie}(1)). This yields a fine-grained framework for analysing epistemic states in which contradictions arise (cf.\ Section~\ref{examples}).

From now on, we consider the AGM$\circ$ system of paraconsistent belief change based on \lfis\ introduced in~\cite{tes.con.rib.2016}, instantiated with the logic \rcbr. All definitions and results below also apply, \emph{mutatis mutandis}, to \rcie. The following notions will be useful:

\begin{definition} \label{unrev} 
Let $K$ be a closed theory in  \rcbr\ (such theories will be called  {\em belief sets} from now on). A sentence $\alpha$ is said to be {\em irrevocable in $K$}, denoted by $I_K(\alpha)$, if either  $\vdash_{\rcbr}\alpha$ or $\circ\alpha \in K$. Otherwise $\alpha$ is said to be {\em revocable in $K$}, which will be denoted by $R_K(\alpha)$.
\end{definition}

Let $Th(\rcbr)$ be the set of closed theories in \rcbr, i.e., the set of belief sets in \rcbr.
The properties of \rcbr\ lead us to the following simplification of {\em ($\div$extensionality)} taken from Definition~5.16 in~\cite{tes.con.rib.2016}:

\begin{definition}[Postulates for extensional  AGM$\circ$ contraction based on \rcbr]\label{def:postcontbola}  An {\em extensional AGM$\circ$ contraction}  over \rcbr\ is a function $\div: Th(\rcbr)\times \mathfrak{Fm} \longrightarrow Th(\rcbr)$ satisfying the following postulates:

\begin{description}
\setlength\itemsep{0em}
\item {\bf ($\div$closure)} $K \div \alpha = Cn(K \div \alpha)$.
\item {\bf ($\div$success)} If $R_K(\alpha)$ then $\alpha \notin K \div \alpha$.
\item {\bf ($\div$inclusion)} $K \div \alpha \subseteq K$.
\item {\bf ($\div$failure)} If $\circ \alpha \in K$ then $K \div \alpha = K$.
\item {\bf ($\div$relevance)} If $\beta \in K \setminus (K \div \alpha)$ then there exists $X$ such that $K \div \alpha \subseteq Cn(X) \subseteq K$ and $\alpha \notin Cn(X)$, but $\alpha \in Cn(X) + \beta$. 
\item {\bf ($\div$extensionality)} If $\alpha \equiv \beta$ then  $K \div \alpha=K \div \beta$.
\end{description}
\end{definition} 

The set of postulates for belief contraction is designed to ensure that the system adheres to fundamental principles of rationality, which govern the dynamics of epistemic change. In particular, they provide the metatheoretical conditions that an operator $\div$ must satisfy in order to count as an extensional contraction on belief sets over \rcbr.

The postulate of \emph{Closure} ensures that the output of a contraction is still a belief set, i.e., it remains closed under logical consequence. The \emph{Success} postulate captures the point of contracting by a sentence: whenever the target sentence is revocable, it is not retained after contraction. In particular, it guarantees that strongly accepted sentences -- those that play a privileged role in the epistemic state within the systems considered here -- are preserved unless the corresponding consistency statement is first retracted (as explained in Remark~\ref{rem-RCie}).  \emph{Inclusion} expresses the requirement of \emph{minimal change} by ruling out the introduction of new beliefs, thus preserving as much of the original belief set as possible. The \emph{Failure} postulate highlights the robustness of \emph{strong acceptance}: sentences regarded as consistent are not removable by contracting them. \emph{Relevance} refines minimal change by requiring that only sentences that matter for deriving the sentences to be retracted are removed, yielding a conservative and context-sensitive change. Finally, \emph{Extensionality} implements the \emph{irrelevance of syntax}: logically equivalent sentences must be treated alike, so that the postulate coincides with its standard AGM counterpart.

It is crucial to recognize that while these postulates are grounded in general principles of rationality, their applicability is context-sensitive. Not all criteria are universally desirable across every scenario of belief change. Conflicts may arise between these principles (as it is discussed by \cite{cost} regarding the very context of Paraconsistent Belief Revision), necessitating a clear prioritization within the epistemological framework to guide the dynamics of belief revision in a consistent and contextually appropriate manner.

\begin{remarks} \label{rem-RCie} (1) In~\cite[Subsection~5.1]{tes.con.rib.2016} the variety of (new) epistemic attitudes available in AGM$\circ$ was analyzed. In particular, when 
both $\alpha$ and $\cons\alpha$ belong to $K$, then $\alpha$ is said to be {\em strongly (or boldly) accepted} in $K$. This implies that $\alpha$ cannot be contracted from $K$.  From this, the revision (and so the expansion) of $K$ by $\neg\alpha$ produce the trivial belief set $K_\bot=\mathfrak{Fm}$. Analogously, if $\neg\alpha$ and $\cons\alpha$ belong to $K$, $\alpha$ is said to be  {\em strongly (or boldly) rejected} in $K$, which implies that $\neg\alpha$ cannot be contracted from $K$.  In this case, the revision (and so the expansion) of $K$ by $\alpha$ produce the trivial belief set $K_\bot$. This fact was explained in~\cite[p.~648]{tes.con.rib.2016} as follows:

\begin{quote}
{\em The consistency of $\alpha$ in $K$ [\ldots]  means that any propositional epistemic attitude about it is irrefutable. If the agent accepts or rejects such belief-representing sentence, $K$ will be non-revisable, respectively, by $\neg\alpha$ and $\alpha$. Furthermore, the sentence will be so entrenched in the epistemic state that to exclude it is not even a possibility.}
\end{quote}

As it was observed in~\cite[Remark~5.8]{tes.con.rib.2016}, for this approach to paraconsistent belief revision, as well as the new epistemic attitudes, to make sense, it is relevant to assume the equivalence between $\cons\alpha$ and  $\cons\neg\alpha$, as well as the equivalence between $\alpha$ and  $\neg\neg\alpha$. From this, boldly rejecting $\alpha$ (i.e., $(\neg\alpha \land \cons\alpha) \in K$) is equivalent to  boldly accepting $\neg\alpha$ (i.e., $(\neg\alpha \land \cons\neg\alpha) \in K$). Analogously, boldly rejecting $\neg\alpha$ (i.e., $(\neg\neg\alpha \land \cons\neg\alpha) \in K$) is equivalent to  boldly accepting $\alpha$ (i.e., $(\alpha \land \cons\alpha) \in K$).  These features fully justify the use of \cbr\ and \rcbr. \\[1mm]
(2) It should be observed that, the fact of $\alpha$ being boldly accepted (or boldly rejected)  in a belief set $K$ under logic \rcbr, {\em does not mean} that this situation will last forever: $\alpha$ (or $\neg\alpha$, respectively) are irrevocable in $K$, but it is not necessarily the case for $\cons\alpha$ itself, since nothing guarantees that $\cons\cons\alpha \in K$. This is a consequence of the fact that neither $(\alpha \land\cons\alpha) \to \cons\cons\alpha$ nor $(\neg\alpha \land\cons\alpha) \to \cons\cons\alpha$ are valid in general in \rcbr, recall Corollary~\ref{cp-not-valid}. Hence, if we contract $K$ with respect to $\cons\alpha$, it is possible to remove  $\alpha$  (or $\neg\alpha$, respectively) afterwards, or even to consider (perhaps  for the sake of argument) both sentences $\alpha$ and $\neg\alpha$, without trivializing, in the new  belief set   $K \div \cons\alpha$. If $\alpha$ is just consistent in $K$ (that is, $\circ\alpha \in K$ but $\alpha \notin K$ and $\neg\alpha \notin K$) the same argument can be applied, and so $\circ\alpha \notin K \div \cons\alpha$ in general.
\\[1mm]
%
%
(3) Let us analyze now the case of \rcie. Note that 
$\cons\cons\alpha$ is a theorem of \rcie, for every formula $\alpha$.  From this, $\cons\cons\alpha \in K$, for every belief set $K$ over \rcie\ and for every formula $\alpha$. But then, if $\div$ is an extensional AGM$\circ$ contraction over \rcie\ it follows, by {\em $\div$failure}, that $K \div \cons\alpha = K$, for every $K$ and every $\alpha$. This feature of \rcie, which implies that $\cons\alpha$ is irrevocable in $K$ for every $K$ and every $\alpha$, has of course deep consequences for the beliefs dynamics based on \rcie: different to \rcbr, if  $\alpha$ is boldly accepted (or boldly rejected) in $K$, then this situation will last forever.
\end{remarks}

To facilitate the proof of Theorem~\ref{cont:entrenchment}, it is worth observing that any contraction satisfying the previously stated postulates also satisfies the \emph{Vacuity} and \emph{Recovery} postulates, as stated below.

\begin{proposition}
If $\div$ is a contraction in \rcbr\ as in Definition~\ref{def:postcontbola} then it satisfies, in addition, the following postulates:
\begin{description}
\item {\bf ($\div$vacuity)} If $\alpha \not\in K$ then $K \div \alpha = K$.
\item  {\bf ($\div$recovery)} $K \subseteq  (K \div \alpha)+ \alpha$.
\end{description}
\end{proposition}
\begin{proof}
{\em $\div$vacuity} is, as in the case of classical contraction over \cpl, a direct consequence of {\em $\div$relevance}  and {\em $\div$inclusion}. In turn, {\em $\div$recovery} follows from {\em $\div$relevance} as in the case of classical contraction over \cpl, as proved in~\cite[p.~114]{hans:99}, taking into account that $(\alpha \to \beta) \to \alpha \vdash_{\rcbr} \alpha$ and \rcbr\ satisfies the deduction metatheorem (DMT), as in the case of \cpl.
\end{proof}

\

\emph{Vacuity} states that if the sentence to be retracted is not in the belief set, then nothing needs to be removed: contracting by it leaves the belief set unchanged. \emph{Recovery}, in turn, says that if one first contracts by a sentence and then expands again by that very sentence, one gets back the original belief set. These were the postulates originally adopted in the AGM framework.  

However, as emphasized in~\cite{recovery}, \emph{recovery} has been widely debated, since it may yield consequences that many authors regard as counterintuitive. For this reason, alternative constraints have been proposed, including \emph{relevance} \cite{relevance} -- the option adopted in the original AGM$\circ$ presentation. Still, for certain background logics (classical logic among them), \emph{vacuity} and \emph{recovery} turn out to be equivalent to \emph{relevance} once combined with the other standard postulates \cite{wassermann}. As shown below, the same equivalence holds for \rcbr. This phenomenon -- sometimes described as the \emph{pertinacity of recovery} \cite{hans:99} -- suggests that \emph{recovery} can arise as a derived feature of AGM-style belief dynamics.

\begin{proposition} \label{relev-prop}
If $\div: Th(\rcbr)\times \mathfrak{Fm} \longrightarrow Th(\rcbr)$ is an operator satisfying
{\em $\div$closure}, {\em $\div$inclusion}, {\em $\div$vacuity} and  {\em $\div$recovery}, then $\div$ satisfies  {\em $\div$relevance}.
\end{proposition}
\begin{proof} It is proven as in the case of classical contraction over \cpl, see for instance~\cite[p.~114--115]{hans:99}, taking into account that, as in the case of \cpl:   $\alpha \vdash_{\rcbr} \beta \to \alpha$; $(\beta \to \alpha) \to \alpha \vdash_{\rcbr} \alpha \vee \beta$;   $\alpha \vee \beta, \alpha \to \beta \vdash_{\rcbr} \beta$;   and \rcbr\ satisfies the deduction metatheorem DMT.
\end{proof}

\

\noindent
Moreover, to formally capture the intuitive idea that the most entrenched beliefs within an epistemic state are precisely those that should be preserved over less entrenched ones -- as reflected by contraction based on epistemic entrenchment -- it becomes necessary, at the level of postulates, to consider additional conditions related to contraction by conjunctions. Indeed, there are two additional postulates in classical AGM, usually called {\em supplementary G\"ardenfors postulates}, concerning the relationship between the contraction with a conjunction $\alpha \land \beta$ and the respective contraction with $\alpha$ and $\beta$, namely:

\begin{description}
\setlength\itemsep{0em}
\item {\bf ($\div$conjunctive overlap)} $(K \div \alpha) \cap  (K \div \beta) \subseteq K \div (\alpha \land \beta)$.
\item {\bf ($\div$conjunctive inclusion)} If $\alpha \notin  K \div (\alpha \land \beta)$ then $K \div (\alpha \land \beta) \subseteq K \div \alpha$.
\end{description}

\

In other words, \emph{Conjunctive Overlap} requires that any belief belonging to both $(K \div \alpha)$ and $(K \div \beta)$ must also belong to $K \div (\alpha \land \beta)$. \emph{Conjunctive Inclusion}, on the other hand, ensures that when contracting $K$ with respect to $\alpha \land \beta$, at least $\alpha$ or $\beta$ must be rejected. This reflects the intuitive idea that giving up a conjunction necessarily involves giving up at least one of its conjuncts. Moreover, if $\alpha$ is the sentence rejected, this occurs because $\alpha$ is less epistemically entrenched than $\beta$, ``making the minimal change required to contract $K$ by $\alpha \land \beta$ closely related to the minimal change needed to reject $\alpha$ itself'' \cite{Gardenfors1988}.

These two additional postulates for contraction, together with the usual ones, characterize exactly the change operators which are defined by means of an epistemic entrenchment. As we shall see in the next section,  the notion of epistemic entrenchment we are proposing for belief change in AGM$\circ$ based on \rcbr\ induces a contraction operator which satisfies a weaker form of the first supplementary postulate above, namely:

\begin{description}
\setlength\itemsep{0em}
\item {\bf ($\div$weak conjunctive overlap)} If either $\vdash_{\rcbr}\alpha$ or  $\vdash_{\rcbr}\beta$ or  ($R_K(\alpha)$ and $R_K(\beta)$) then   $(K \div \alpha) \cap  (K \div \beta) \subseteq K \div (\alpha \land \beta)$.
\end{description}

Such weakening should not come as a surprise, given that we are now dealing with new epistemic attitudes, in which strongly accepted beliefs are regarded as the most entrenched ones -- alongside tautologies, which together constitute the so-called irrevocable belief-representing sentences. An intuitive way to understand \emph{weak conjunctive overlap} is to observe that the conclusion of the postulate (namely, its traditional version) fails only if both conjuncts $\alpha$ and $\beta$ are not tautologies, yet at least one of them is regarded as consistent and, therefore, irrevocable. Now, if one of the conjuncts (say, $\alpha$) is a tautology, then it trivially follows that $(K \div \alpha) \cap (K \div \beta) \subseteq K \div (\alpha \land \beta)$ -- regardless of the epistemic status of $\beta$ -- since, in such case, $(K \div \alpha)=K$ and $\beta$ is logically equivalent to $\alpha\land\beta$ (and vice versa).

Before moving on to the next section -- where we define contraction in AGM{$\circ$} by means of epistemic entrenchment -- it is worth stressing a distinctive feature of the AGM{$\circ$} model: it supports alternative ways of defining revision, beyond the standard one induced by the Levi identity. Usually, to revise $K$ by $\alpha$ one first contracts $K$ by $\neg\alpha$, so that $\alpha$ can be added without contradiction; the resulting belief set is then expanded by $\alpha$.

Paraconsistent belief revision, as developed in~\cite{tes.con.rib.2016}, also admits alternative constructions. In particular, the \emph{Reverse Levi Identity} reverses the order of the two steps: one first expands $K$ by $\alpha$ and then contracts the resulting (possibly contradictory) epistemic state by $\neg\alpha$. This yields what is called \emph{external revision} (in contrast with the usual \emph{internal revision}). More generally, \emph{semi-revision} can be understood as a refinement of external revision in which consolidation is achieved not necessarily by contracting $\neg\alpha$, but by contracting the explicit contradiction(s) that arise after expansion. As a consequence, $\alpha$ need not be preserved: the consolidation step may retract whatever sentences are required to eliminate those contradiction(s), possibly including $\alpha$ itself.

Such distinct revisions -- where a contradictory state may occur as an intermediate stage -- are often discussed in the context of belief base models, in which belief states are not deductively closed and contradictions need not lead to triviality. A distinctive contribution of~\cite{tes.con.rib.2016} is to articulate these operators within belief sets closed under paraconsistent logics, namely the \lfis. The present paper takes a further step by formulating them in terms of epistemic entrenchment.

\section{Epistemic entrenchment in extensional AGM$\circ$ over  \rcbr}
 
Epistemic entrenchment, introduced by Gärdenfors (see~\cite{Gardenfors1988}), is based on the idea that contractions in a belief set $K$ should be guided by an ordering of sentences according to their epistemic importance. This ordering, determined prior to contraction, is influenced by factors such as explanatory power, overall informational value, and probabilities, depending on the context of application. As noted by Rott \cite{Rott1992}, a simpler interpretation is that, if $\alpha$ is less entrenched than $\beta$, it means ``it is easier to discard $\alpha$ than $\beta$.''

The formal tool used by Gärdenfors to represent epistemic entrenchment is a binary relation that satisfies specific postulates.

In this section, the notation introduced in Definition~\ref{unrev} will be frequently used. The postulates for epistemic entrenchment for the classical AGM can be adapted to \rcbr\ as follows:

\begin{definition} \label{def-EE}
Let $K$ be a belief set in \rcbr. A relation ${\leq} \subseteq K \times K$ is an {\em epistemic entrenchment for $K$} if it satisfies the following properties:

\begin{description}
\setlength\itemsep{0em}
\item {\bf (EE1)} If $\alpha \leq \beta$ and $\beta \leq \gamma$, then $\alpha \leq \gamma$ \hspace{1cm}  {\em (transitivity)}.
\item {\bf (EE2)} If $\alpha \vdash_{\rcbr} \beta$ or ${\circ}\beta \in K$, then $\alpha \leq \beta$ \hspace{1cm}  {\em (dominance)}.
\item {\bf (EE3)} Either $\alpha \leq \alpha \land \beta$ or $\beta \leq \alpha \land \beta$ \hspace{1cm}  {\em (conjunctiveness)}.
\item {\bf (EE4)} If $K \neq K_\bot$, then $\alpha \notin K$ iff $\alpha \leq \beta$ for all $\beta$ \hspace{1cm}  {\em (minimality)}.
\item {\bf (EE5)} If $\beta \leq \alpha$ for all $\beta$, then  $I_K(\alpha)$ \hspace{1cm}  {\em (maximality)}.
\end{description}
\end{definition}

\noindent As usual, we will write $\alpha < \beta$ to denote that $\alpha \leq \beta$ and $\beta \not\leq \alpha$.

It is worth emphasizing that, in our framework, the most significant deviations from the classical entrenchment paradigm appear in postulates (EE2) and (EE5). (EE1), as usual, ensures that epistemic entrenchment is transitive. Traditionally, (EE2) states that if $\alpha$ entails $\beta$, then $\alpha$ must be at most as entrenched as $\beta$, since one cannot retract $\beta$ without also giving up $\alpha$. In our framework, this condition is reinforced to reflect the epistemic status of strongly accepted sentences in the AGM$\circ$ system. Specifically, if $\circ\beta\in K$, then $\beta$ is considered a consistent belief -- an indicator that the agent is committed not just to $\beta$, but to its stability under revision. Thus, any sentence $\alpha$ must be at most as entrenched as $\beta$, even if there is no formal derivability from $\alpha$ to $\beta$. This extension aligns with the idea that strongly accepted beliefs (i.e., the acceptance of both $\beta$ and $\circ\beta$ play a foundational role in the present paradigm of belief dynamics, and should occupy the top levels of the entrenchment hierarchy. (EE3) captures the intuition that discarding a conjunction requires discarding at least one of its conjuncts. Accordingly, a conjunction should be at least as entrenched as one of its components. (EE4) stipulates that the least entrenched beliefs are precisely those absent from the belief set. (EE5) characterizes the top of the entrenchment hierarchy: if a belief $\alpha$ is more entrenched than every other belief in $K$, then it must be irrevocable -- either because it is a logical theorem or because its consistency is explicitly asserted in the belief set (i.e., $\circ\alpha\in K$). Together with (EE2), this condition plays a crucial role in establishing the full equivalence between maximal entrenchment and irrevocability. This result is formally proven in Proposition \ref{prop-EE}.

\begin{proposition} \label{prop-EE} Let $K$ be a belief set in \rcbr, and let ${\leq} \subseteq K \times K$ be an  epistemic entrenchment for $K$ Then:
$\beta \leq \alpha$ for all $\beta$ iff $I_K(\alpha)$.
\end{proposition}
\begin{proof} 
It is  an immediate consequence of {\em (EE2)} and {\em (EE5)}. 
Indeed, the `only if' part follows from {\em (EE5)}. Now, observe that if  $\vdash_{\rcbr} \alpha$ then $\beta\vdash_{\rcbr} \alpha$ for every $\beta$.  Hence, if $I_K(\alpha)$ then  $\beta \leq \alpha$ for all $\beta$, by {\em (EE2)}.
\end{proof}

\begin{remarks} \label{rem-EE}
(1) Observation 2.75 in~\cite{hans:99} holds in our framework. Namely: Let $\leq$ be a relation that satisfies  {\em (EE1)}. Then it satisfies:
\begin{enumerate}
\item $x \leq y$ and $y < z$ implies $x < z$.
\item $x < y$ and $y \leq z$ implies $x < z$.
\item $x < y$ and $y < z$ implies $x < z$ (quasi-transitivity).
\end{enumerate}
(2) Observation 2.48 in~\cite{hans:99} also holds in our framework. Namely: Let $\leq$ be a relation that satisfies  {\em (EE1)},  {\em (EE2)} and  {\em (EE3)}. Then it satisfies:\\[1mm]
Either $\alpha \ \leq \beta$ or $\beta \leq \alpha$ (connectivity).\\[1mm]
(3) Observation 2.92 in~\cite{hans:99} also holds in our framework. Namely: Let $\leq$ be a relation that satisfies  {\em (EE1)} and  {\em (EE2)}. Then it also satisfies:\\[1mm]
If $\alpha \equiv \alpha'$ and  $\beta \equiv \beta'$ then $\alpha \leq \beta$ iff $\alpha' \leq \beta'$ (intersubstitutivity).\\[1mm]
(4) Observation 2.93 in~\cite{hans:99} also holds in our framework. Namely: Let $\leq$ be a relation that satisfies  {\em (EE1)},  {\em (EE2)} and  {\em (EE3)}. Then it also satisfies:
\begin{enumerate}
\item If $\alpha \land \beta \leq \delta$, then either $\alpha \leq \delta$ or $\beta \leq \delta$.
\item  If $\delta < \alpha$ and $\delta < \beta$, then  $\delta < \alpha \land \beta$  (conjunction up).
\end{enumerate}
The proof of items (2)-(4) is analogous to the classical case, by observing that our version of {\em (EE2)} is stronger than the classical one.
\end{remarks}

\

\noindent Epistemic entrenchment and contraction can be related in our framework as follows:

\begin{description}
\setlength\itemsep{0em}
\item $(C\leq)$ \ \ $\alpha \leq \beta$ \ iff \ $\alpha \notin K \div (\alpha \land \beta)$ or $I_K(\alpha \land \beta)$.
\item $(G\div)$ \ \ $\beta \in K\div\alpha$ \ iff \ $\beta \in K$ and either $\alpha < \alpha \vee \beta$ or $I_K(\alpha)$.
\end{description}

\

\noindent Indeed, given an epistemic entrenchment defined for every belief set $K$ in \rcbr, it is possible to define an AGM$\circ$ contraction operator $\div$  by using $(G\div)$:

\begin{theorem}\label{cont:entrenchment}
For every belief set $K$ in \rcbr\ let $\leq$ be an epistemic entrenchment ordering for $K$ as in Definition~\ref{def-EE}. Then, the function $\div: Th(\rcbr)\times \mathfrak{Fm} \longrightarrow Th(\rcbr)$ obtained from $(G\div)$, for each $K$, is an operator satisfying the postulates for extensional  AGM$\circ$ contraction based on \rcbr\ (recall Definition~\ref{def:postcontbola}) plus the supplementary ones: {\em ($\div$weak conjunctive overlap)} and  {\em ($\div$conjunctive inclusion)}. In addition, condition $(C\leq)$ holds, for every $K$.
\end{theorem}
\begin{proof} 
Define $\div: Th(\rcbr)\times \mathfrak{Fm} \longrightarrow Th(\rcbr)$ by using $(G\div)$ from each $\leq$ over $K$, for $K \in Th(\rcbr)$. To be more precise, given $K$, $\alpha$ and $\leq$ over $K$ let\\

$K \div\alpha := \{\beta \in K \ : \ \mbox{either $\alpha < \alpha \vee \beta$ or $I_K(\alpha)$} \}\hspace*{3cm}(def\div)$\\

\noindent Let us prove that $\div$ satisfies the required properties. Our proof is an adaptation of the one for the classical case found in~\cite{hans:99}, pages 188--190.\\[2mm]
{\bf ($\div$closure):} If $I_K(\alpha)$ then $K\div \alpha=K$, by construction, which is a closed set. Now, assume that $R_K(\alpha)$, and suppose that $K\div\alpha \vdash_{\rcbr}\delta$. We need to show that $\delta \in K$ and $\alpha < \alpha \vee \delta$. \\[1mm]
Case 1: Suppose that $K\div\alpha \neq\emptyset$.
By definition of $\vdash_{\rcbr}$, either $\vdash_{\rcbr}\delta$, or there exists $\beta_1,\ldots,\beta_n \in K\div\alpha$ such that $(\beta_1 \land \ldots \land \beta_n) \to \delta$ is a theorem of \rcbr. In the first case, $\delta \in K$. In the second case, by $(def\div)$, $\beta_1,\ldots,\beta_n \in K$ and so  $K \vdash_{\rcbr}\delta$. But $K$ is a closed theory, hence $\delta \in K$ also in this case. In turn, since $\beta_i \in K\div\alpha$ and $R_K(\alpha)$ then, by $(def\div)$, $\alpha < \alpha \vee \beta_i$, for every $1 \leq i \leq n$. By Remarks~\ref{rem-EE}(4) it follows that $\alpha < (\alpha \vee \beta_1) \land \ldots \land (\alpha \vee \beta_n)$. By Remarks~\ref{rem-EE}(3), $\alpha < \alpha \vee (\beta_1 \land \ldots \land \beta_n)$. But $\beta_1 \land \ldots \land \beta_n \vdash_{\rcbr} \delta$, hence $ \alpha \vee(\beta_1 \land \ldots \land \beta_n) \vdash_{\rcbr} \alpha \vee \delta$ and so $ \alpha \vee(\beta_1 \land \ldots \land \beta_n) \leq \alpha \vee \delta$, by (EE2). By (EE1) and Remarks~\ref{rem-EE}(1.2), $\alpha < \alpha \vee \delta$.\\[1mm]
Case 2: Assume that  $K\div\alpha =\emptyset$. This means that $\vdash_{\rcbr}\delta$, hence  $\delta \in K$. Since $\vdash_{\rcbr} \alpha \vee \delta$, it follows by (EE2) that $\beta \leq \alpha \vee \delta$ for every $\beta$. Given that $R_K(\alpha)$ we infer from (EE5)  that $\beta \not\leq \alpha$, for some $\beta$. By  Remarks~\ref{rem-EE}(2), $\alpha < \beta$. Since  $\beta \leq \alpha \vee \delta$ then $\alpha <  \alpha \vee \delta$, by Remarks~\ref{rem-EE}(1.2).\\[2mm]
{\bf ($\div$sucess):} Assume that $R_K(\alpha)$. By (EE2), $\alpha \vee \alpha \leq \alpha$, hence $\alpha \not<\alpha \vee \alpha$. By $(def\div)$, $\alpha \not\in K \div\alpha$.\\[2mm]
{\bf ($\div$inclusion):} It is immediate, by $(def\div)$.\\[2mm]
{\bf ($\div$failure):} It is also immediate, by $(def\div)$.\\[2mm]
{\bf ($\div$vacuity):} Let $\alpha \notin K$ (hence $K \neq K_\bot$). By $(def\div)$, $K \div \alpha \subseteq K$. We need  to prove that $K \subseteq K \div \alpha$, so let $\beta \in K$. By (EE4), $\beta \not\leq \gamma$ for some $\gamma$. By  Remarks~\ref{rem-EE}(2), $\gamma < \beta$, while $\beta \leq \alpha \vee \beta$ by (EE2). Since $\alpha \notin K$ then $\alpha \leq \gamma$, by (EE4). By  Remarks~\ref{rem-EE}(1), $\alpha < \alpha \vee \beta$. By $(def\div)$, $\beta \in K \div\alpha$.\\[2mm]
{\bf ($\div$recovery):} Suppose that $\beta \in K$. It is enough proving that $\alpha \to \beta \in K \div \alpha$. Since $\beta \vdash_{\rcbr} \alpha \to \beta$ and $K$ is closed, $\alpha \to \beta \in K$.\\[1mm]
Case 1: Assume that $I_K(\alpha)$. Then, $\alpha \to \beta \in K \div \alpha$, by $(def\div)$. \\[1mm]
Case 2: Suppose that $R_K(\alpha)$. By (EE5), there is some $\gamma$ such that $\gamma \not\leq\alpha$. By  Remarks~\ref{rem-EE}(2), $\alpha < \gamma$. Since $\vdash_{\rcbr} \alpha \vee (\alpha \to \beta)$, it follows by (EE2) that $\gamma \leq \alpha \vee (\alpha \to \beta)$.  By  Remarks~\ref{rem-EE}(1.2), $\alpha < \alpha \vee (\alpha \to \beta)$. By $(def\div)$, $\alpha \to \beta \in K \div \alpha$.\\[2mm]
{\bf ($\div$relevance):} It follows from Proposition~\ref{relev-prop} and the properties of $\div$ already proven.\\[2mm]
{\bf ($\div$extensionality):} Assume that $\alpha \equiv \beta$. It is enough proving that $K\div\alpha \subseteq K \div\beta$. Thus, let $\delta \in K\div\alpha$. By $(def\div)$, $\delta \in K$ is such that either $\alpha < \alpha \vee \delta$ or $I_K(\alpha)$. By Remark~\ref{rem-RCbr}, $I_K(\alpha)$ iff $I_K(\beta)$. On the other hand, $\alpha \equiv \beta$ implies that $\alpha \vee \delta \equiv \beta \vee \delta$. By  Remarks~\ref{rem-EE}(3), $\alpha < \alpha \vee \delta$ iff $\beta < \beta \vee \delta$. This means that  $\delta \in K\div\beta$, by $(def\div)$.\\[2mm]
{\bf ($\div$weak conjunctive overlap):} Suppose that  $\vdash_{\rcbr}\alpha$. Hence, $\beta \equiv (\alpha \land \beta)$ and so, by {\em $\div$extensionality} (which was already proven), we infer that $K \div\beta= K \div(\alpha \land \beta)$. From this,  $(K \div \alpha) \cap  (K \div \beta) \subseteq K \div (\alpha \land \beta)$. If $\vdash_{\rcbr}\beta$ then, by a similar reasoning, we prove that $(K \div \alpha) \cap  (K \div \beta) \subseteq K \div (\alpha \land \beta)$. Now, suppose that $R_K(\alpha)$ and $R_K(\beta)$. Let $\delta \in (K \div \alpha) \cap  (K \div \beta)$.  By (EE2), $\alpha \land \beta \leq \alpha$. Since $\delta \in K \div \alpha$ then, by $(def\div)$, $\alpha < \alpha \vee \delta$.  By  Remarks~\ref{rem-EE}(1.1), $\alpha \land \beta < \alpha \vee \delta$. Analogously it is proven that $\alpha \land \beta < \beta \vee \delta$. By  Remarks~\ref{rem-EE}(4.2), $\alpha \land \beta < (\alpha \vee \delta) \land (\beta \vee \delta)$. Using  Remarks~\ref{rem-EE}(3) it follows that $\alpha \land \beta < (\alpha \land \beta) \vee \delta$. Given that $\delta \in K$, it follows that $\delta \in K \div (\alpha \land \beta)$,  by $(def\div)$.\\[2mm]
{\bf ($\div$conjunctive inclusion):} Assume that $\alpha \notin K \div (\alpha \land \beta)$. Then, $\nvdash_{\rcbr}\alpha$.  By $(def\div)$, $R_k(\alpha \land\beta)$.\\[1mm]
Case 1: $\alpha \notin K$. Then $\alpha \land\beta \notin K$ and so, by {\em $\div$vacuity} (which was already proven), $K \div (\alpha \land \beta) = K = K \div \alpha$.\\[1mm]
Case 2: $\alpha \in K$. Let $\delta \notin K \div\alpha$. We want to prove that $\delta\notin K \div (\alpha \land \beta)$. If $\delta \not\in K$, this is immediate. Now, assume that $\delta \in K$. By (EE2), $(\alpha \land \beta) \vee \delta \leq \alpha \vee \delta$. Since $\delta \in A$ and $\delta \notin  K \div\alpha$ then, by $(def\div)$, $\alpha \not< \alpha \vee \delta$. By  Remarks~\ref{rem-EE}(2), $\alpha \vee \delta \leq \alpha$. Since $\alpha \in A$ and $\alpha \notin  K \div(\alpha \land \beta)$ then, by a similar reasoning, it follows that $(\alpha \land \beta) \vee \alpha \leq \alpha \land \beta$. By  Remarks~\ref{rem-EE}(3), $\alpha \leq \alpha \land \beta$. From $(\alpha \land \beta) \vee \delta \leq \alpha \vee \delta$, $\alpha \vee \delta \leq \alpha$, and $\alpha \leq \alpha \land \beta$ it follows that $(\alpha \land \beta) \vee \delta \leq \alpha \land \beta$, by (EE1). That is, $\alpha \land \beta \not<  (\alpha \land \beta) \vee \delta$. By $(def\div)$ it follows that $\delta \notin K \div (\alpha \land \beta)$, since $R_k(\alpha \land\beta)$.\\[2mm]
{\bf Condition $(C\leq)$:} Suppose first that $\alpha \leq \beta$. Assuming that $\alpha \in K \div (\alpha \land \beta)$, we need to show that $I_K(\alpha \land \beta)$. By $(def\div)$, $\alpha \in K$ and either $\alpha \land\beta < (\alpha \land \beta) \vee \alpha$ or $I_K(\alpha \land \beta)$. By (EE3), either $\alpha \leq \alpha \land \beta$ or $\beta \leq \alpha \land \beta$. If $\beta \leq \alpha \land \beta$ then, from $\alpha \leq \beta$ and (EE1), it follows that $\alpha \leq \alpha \land \beta$. That is, in both cases $\alpha \leq \alpha \land \beta$. By  Remarks~\ref{rem-EE}(3), $(\alpha \land \beta) \vee \alpha \leq \alpha \land \beta$, which implies that $\alpha \land \beta \not< (\alpha \land \beta) \vee \alpha$. This shows that $I_K(\alpha \land \beta)$, as desired. Conversely, assume that either $\alpha \notin K \div (\alpha \land \beta)$ or $I_K(\alpha \land \beta)$. We want to prove that $\alpha \leq \beta$.\\[1mm]
Case 1: Suppose that $I_K(\alpha \land \beta)$. By  Proposition~\ref{prop-EE}, $\delta \leq \alpha \land \beta$ for every $\delta$. In particular, $\alpha \leq \alpha \land \beta$. Since $\alpha \land \beta \leq \beta$, by (EE2), then $\alpha \leq \beta$, by (EE1).\\[1mm]
Case 2: Suppose that $R_K(\alpha \land \beta)$. Then, $\alpha \notin K \div (\alpha \land \beta)$. By $(def\div)$,  either $\alpha \not\in K$ or $\alpha \land \beta \not< (\alpha \land \beta) \vee \alpha$ or $R_K(\alpha \land \beta)$. This implies that either  $\alpha \not\in K$ or $\alpha \land \beta \not< (\alpha \land \beta) \vee \alpha$. If $\alpha \not\in K$ then $K \neq K_\bot$ and so $\alpha \leq \beta$, by (EE4). Finally, if $\alpha \land \beta \not< (\alpha \land \beta) \vee \alpha$ then $(\alpha \land \beta) \vee \alpha \leq \alpha \land \beta$, by Remarks~\ref{rem-EE}(2). By (EE2), $\alpha \leq (\alpha \land \beta) \vee \alpha$ and $\alpha \land \beta \leq \beta$, hence, by (EE1), $\alpha \leq \beta$ also in this case.\\

This concludes the proof.
\end{proof} 

\

\subsection{Applications, Connections and Examples}\label{examples}
G\"{a}rdenfors \cite{Gardenfors1988} distinguishes two main strands in the early development of epistemic entrenchment: the \emph{information-theoretic} approach and the \emph{paradigm} approach. On the former, entrenchment is tied to the informational value of sentences: the more informative a sentence is, the more entrenched it should be, since information is valuable and it is rational to minimize information loss. How exactly such value should be measured -- and even which notion of information is relevant -- remains unsettled. Still, the literature offers several proposals, including probabilistic and other ranking-based models \cite{ferme.hansson}. For our purposes, it suffices to assume that, insofar as information loss can be assessed, a sentence's degree of epistemic entrenchment can be identified with the amount of information that would be lost were that sentence to be retracted from the belief state.

In a paraconsistent setting, the information-theoretic paradigm yields a trade-off between the minimality desideratum (i.e. the minimization of information loss) and consistency, given the rational possibility of contradictions (that is, the logical possibility of contradictory yet non-trivial epistemic states). In this sense, consistency becomes an additional constraint on minimal change, which may come at the price of further information loss -- this tension is precisely what we call the the \emph{cost of consistency} \cite{cost}.

The second approach is centered on a structural view of scientific theories: the core of a theory contains a set of fundamental assumptions that are protected from immediate revision, while the extended core ``contains a set of special laws that may be introduced, tested, and also rejected, while the core is kept constant''\cite{Gardenfors1988}. In the approach advanced in \cite{Testa2014,tes.con.rib.2016} and extended in the present paper, this core is identified with the strongly accepted beliefs: sentences that are assigned the formal status of consistency and, together with the tautologies, are treated as \emph{irrevocable} (cf. definition \ref{unrev}).

As it is pointed out by Levi \cite{levi}, epistemic entrenchment is also partly fixed by pragmatic considerations -- such as simplicity, explanatory power, and the subject matter at stake -- and is therefore context-dependent. This flexibility is precisely what allows the same formal apparatus to be applied across different domains. Some concrete applications include changes in legal codes, where contraction and revision correspond, respectively, to \emph{derogation} and \emph{amendment} \cite{am}. A simple illustration of the derogation problem is the following \cite{hilpinen}.

\begin{example}\label{deontic1}
Suppose that a father has commanded that:\\
(a) The children may watch TV only if they eat dinner.\\
(b) The children may eat their dinner only if they do their homework.

A consequence of these norms is that:\\
(c) The children may watch TV only if they do their homework.

Suppose now that the father adds the following permission:\\
(d) Today, the children may watch TV without doing their homework.
\end{example}

Let $K$ be the normative code consisting of (a) and (b), understood as a set of sentences closed under $\rcbr$. Although the example can be naturally read in deontic terms, the relevant conflict arises already at the propositional non-modal level. Accordingly, we may represent (a) and (b) as $T \to D$ and $D \to H$ (where $T$ abbreviates ``the children may watch TV'', $D$ ``the children eat dinner'', and $H$ ``the children do their homework''). Thus, $T\to H\in Cn(\{T\to D, D\to H\})$.

Now, if the father adds $T\wedge\neg H$, it is natural to read it as a \emph{temporary exception} meant to hold \emph{today}, and thus as something that should override the derived restriction $T\to H$ only in that context. A natural way of representing this input is to first \emph{expand} $K$ with $T\wedge \neg H$, thereby allowing an intermediate conflicting state $K'$ where $T\to D$ and $D\to H$ still hold, and thus $H\wedge\neg H\in K'$. This is precisely where epistemic entrenchment provides a guiding structure. Let $\leq_{K'}$ be an epistemic entrenchment preorder associated with $K'$. Intuitively, the context ``today'' is reflected by taking the new input $T\wedge\neg H$ to be highly entrenched, relative to the background commitments supporting $T\to H$. Accordingly, a consolidated, non-contradictory state $K''$ can be obtained as the result of $K'\div(H\wedge\neg H)$ that retracts at least one of $T\to D$ or $D\to H$, as determined by $\leq_{K'}$. This can be regarded as an implicit contextual derogation (for today). However, it is equally clear that, on any other day, $T\to D$ and $D\to H$ have priority over $T\wedge \neg H$; accordingly, the consolidated output $K'\div(H\wedge\neg H)$ derogates the temporary exception $T\wedge\neg H$.

In more realistic cases, such an entrenchment ordering can be understood as capturing \emph{hierarchies of regulation} \cite{am}, i.e.\ priority structures according to which normative conflicts are resolved by (possibly implicit) derogations that respect the hierarchy. At the same time, much work in AI \& Law, defeasible reasoning and argumentation-based models emphasizes that legal reasoning is often carried out under inconsistent and evolving information: competing norms and exception-clauses may coexist at intermediate stages, and the role of priorities is precisely to manage such conflicts rather than to guarantee that they never arise in the first place \cite{Prakken1997,testa-con-mestrado}. This makes it especially natural to treat contextual exceptions via the two-step dynamics highlighted above: first, an \emph{expansion} that registers the case-specific information (thereby allowing the conflict to be marked), and then a \emph{consolidation} (defined as a contraction-driven step that retracts enough commitments to remove the \emph{explicit contradictions} marked at the intermediate stage). In this sense, one of the gains of the present paraconsistent paradigm is to accommodate such temporary contradictions in a controlled way, while still delivering a consolidated output when required.

Furthermore, in many domains, some norms may be treated as non-negotiable (e.g., constitutional or higher-order rules). In our setting, this is naturally captured by the distinguished set of \emph{irrevocable} commitments, i.e.\ norms that are maximally entrenched and must be preserved from derogations -- unless their irrevocability is itself retracted, in the sense of Remark~\ref{rem-RCie}(2).

Another example concerns case-specific exceptions (or \emph{defeaters}) in scientific reasoning. An anomalous datum is often addressed by diagnosing the failure of an auxiliary condition -- rather than by giving up core assumptions. The next example illustrates this in a concrete evidential scenario.

\begin{example}\label{doping1}
Consider an anti-doping agency evaluating a laboratory report. Let:\\
(a) The assay is reliable in the present case (no relevant contamination, cross-reactivity, or chain-of-custody failure). ($R$)\\
(b) If the assay is reliable and the result is positive, then a violation should be concluded. ($(R \wedge P)\to V$)\\
(c) The result is positive. ($P$)

Now suppose the agency receives additional information about this particular case:\\
(d) The athlete used a legal supplement later shown to be contaminated in a way that can trigger this assay. ($C$)\\
(e) If (d) holds, then the reliability assumption (a) fails in this case. ($C\to\neg R$)
\end{example}

Let $K$ be the background state consisting of (a)--(c), understood as a set of sentences closed under $\rcbr$. Thus, $K=Cn(\{R,(R\wedge P)\to V,P\})$, and in particular $V\in K$ (a violation should be concluded). If the agency receives (d) and (e), a natural first step is to consider the expansion
$K'=K+\{C,\,C\to\neg R\}$, from which $\neg R$ is also supported. However, the intermediate conflict $R\wedge\neg R$ does not by itself determine what the agency should conclude: whether a violation is to be concluded should be assessed only after a contraction-driven \emph{consolidation} step, guided by $\le_{K'}$. If the new information $C$ and $C\to\neg R$ is regarded as trustworthy, then $R\le_{K'}\neg R$, so that $R\notin K'\div(R\wedge\neg R)$ and hence $R\wedge P$ is not supported; thus no violation can be concluded.

However, several other scenarios can be considered, and the point of the example is precisely that the outcome is bounded by the entrenchment profile assumed at the intermediate stage. For instance, if $C$ is not regarded as trustworthy, thus $\neg R \le_{K'} R$; accordingly, the consolidation step may preserve $R$ and retract (at least) the defeater information, yielding again support for $R\wedge P$ and hence for $V$ -- unless $(R\wedge P)\to V$ itself is argued to be even less reliable than $C$ (e.g.\ in light of recurrent anomalies), in which case consolidation may instead aim at retracting that policy.

By contrast, if the policy is taken to be a ``hard'' scientific fact -- e.g.\ if $\circ((R\wedge P)\to V)$ holds, in the sense that this conditional belongs to the set of irrevocable (hence maximally entrenched) sentences -- then consolidation cannot proceed by retracting it, and the burden of restoring a consolidated output must fall elsewhere (typically on $R$ or on the defeater side, depending on $\le_{K'}$). In all such cases, what epistemic entrenchment over $\rcbr$ adds is a disciplined way of making these ranking-driven choices explicit.

Overall, the preceding examples bring out two complementary reasons for working with epistemic entrenchment over $\rcbr$. First, our framework treats intermediate, potentially conflictive stages as methodologically central: contextual exceptions, defeaters and competing constraints are naturally registered by expansion, often yielding explicit contradictions, and only then removed (when required) by a contraction-driven consolidation. In this sense, the expansion--consolidation pattern (hence semi-revision) is not a mere artifact of the formalism, but a faithful representation of the dynamics presupposed by our motivating applications.

Second, the algebraic machinery built into $\rcbr$ is not just technical overhead: by guaranteeing
\emph{replacement}, it ensures that both entrenchment and the induced contraction are driven by informational content rather than by syntactic presentation. Informationally equivalent formulations of a belief-representing sentence (a commitment or an exception/defeater) receive the same comparative standing, so that the dynamics is stable under representational variation and the entrenchment preorder is well-defined at the level that matters for applications.

\section{Conclusion}

This paper has proposed a technically grounded and conceptually motivated framework for belief revision within the paraconsistent setting of the \emph{Logics of Formal Inconsistency} (\lfis). By systematizing the main results concerning \cbr\ -- first introduced by Testa, Coniglio, and Ribeiro~\cite{tes.con.rib.2016} -- and advancing its self-extensional extension, \rcbr\ -- which builds on the techniques of \emph{Boolean Algebra with LFIs Operators} (BALFIs), developed by Carnielli, Coniglio, and Fuenmayor~\cite{car.con.fue.2022} -- we have addressed structural limitations in previous \lfi-based approaches to Paraconsistent Belief Revision, particularly the absence of the \emph{replacement property}, as emphasized in~\cite{tfgr}. With this, we have succeeded in formally characterizing epistemic entrenchment within a robust paraconsistent framework.

We would like to stress the use of the Chinese Remainder Theorem (CRT) in Section~\ref{sect:countermodel}. This allowed us to state fundamental algebraic properties of a BALFI model for \rcbr, contributing to a better understanding of the semantic aspects of paraconsistent belief revision systems based on \lfis\ with replacement.

The main contributions of this paper are twofold. From the perspective of paraconsistency -- and \lfis\ in particular -- our results develop the interpretation of paraconsistency in terms of contradictory yet nontrivial epistemic states of rational agents~\cite{Testa2014}. By showing how the acceptance of contradictions enables meaningful and well-founded models of belief dynamics, paraconsistent logics themselves are enriched with a principled interpretation and an intuitive understanding grounded in rational agency. In this context, epistemic entrenchment plays a central role by supporting a refined reading of the formal consistency operator: now understood as expressing an epistemic attitude that grounds strong forms of acceptance and rejection. Within the epistemic entrenchment paradigm, these attitudes are further clarified by the ranking of beliefs -- where strongly accepted beliefs are characterized as irrevocable, alongside tautologies.

From the perspective of the Theory of Belief Revision, our framework expands the theoretical landscape of belief dynamics in non-classical logics~\cite{wassermann,ribeiro}, demonstrating how \lfis\ allow for novel and coherent constructions that broaden the classical AGM framework. These results not only extend the applicability of entrenchment-based approaches, but also shed light on how core principles of rationality are themselves shaped by the logical framework in which they are formulated. In this sense, when traditionally assumed postulates of rational belief change are challenged or redefined within non-classical settings, we come to see that our very notion of rationality is not fixed or logic-independent, but rather deeply intertwined with the logical structures we take for granted.

From a philosophical standpoint, this latter question permeates several directions for future development. One particularly compelling avenue lies in the domain of \emph{Multiagent Belief Change}: how can the beliefs of multiple agents be combined, especially in collective decision-making contexts? This problem is deeply connected to Belief Revision theory, including operations such as \emph{belief merging} and its ties to \emph{Social Choice Theory}, as studied by Konieczny and Pino Pérez~\cite{kpp}, and further explored within the framework of Game Theory by Booth and Meyer~\cite{bm}, who investigated equilibrium-based models of belief merging.

However, such models often face impossibility results -- typically arising from tensions between certain social desiderata and the rationality postulates required to resolve conflicting information. We argue that reexamining these postulates within a non-classical logical framework—such as that offered by \lfis\ -- may provide a promising path forward. In particular, the exploration of the enriched epistemic landscape enabled by \textbf{RCbr}, together with its technically well-founded notion of entrenchment, within the context of the aforementioned applications and theoretical connections, constitutes an important direction of ongoing research.

Another independent yet related line of investigation concerns the study of changes in the strength of beliefs, as explored, for instance, in the so-called improvement operations~\cite{213}. These operations do not necessarily satisfy the Success postulate, although they aim to increase an agent’s entrenchment regarding a given belief-representing sentence. The similarities between such operations and the epistemic acceptance of consistency -- as introduced in the AGM$\circ$ framework -- warrant further investigation. Specifically, revising a belief set in order to accept the consistency of a sentence (rather than the sentence itself) may be understood as a shift in the strength of that sentence’s epistemic status. This connection may offer a more substantial interpretation of \lfis\ within the proposed epistemic setting, enriching the philosophical understanding of formal consistency and its role in rational agency.

\

\noindent{\large \bf Acknowledgments:}
This research was supported by the S\~ao Paulo Research Foundation (FAPESP, Brazil), through the Thematic Project RatioLog \#2020/16353-3. The first author also acknowledges support from the National Council for Scientific and Technological Development (CNPq, Brazil), through the individual research grant \#309830/2023-0. The second author also acknowledges suport from FAPESP through the Visiting Researcher Grant - International \#2022/03862-2.  The third author acknowledges support from FAPESP through the Research Internship Abroad (BEPE) grant \#2017/10836-0.

\end{document}